\documentclass[11pt,a4paper]{article}
\usepackage{amssymb, amsmath, amscd, amsfonts, amsthm}
\usepackage{graphicx}
\usepackage{a4wide}
\usepackage{tensor}
\usepackage{mathrsfs}
\usepackage{hyperref}
\usepackage{color}

\definecolor{darkred}{rgb}{0.8,0.1,0.1}
\hypersetup{
     colorlinks=true,         
     linkcolor=darkred,
     citecolor=blue,
}

\usepackage{comment}
\usepackage[margin=2.35cm]{geometry}

\newcommand{\supp}{{\rm supp}}

\theoremstyle{plain}
\newtheorem{theo}{Theorem}[section]
\newtheorem{lem}[theo]{Lemma}
\newtheorem{propo}[theo]{Proposition}

\theoremstyle{definition}
\newtheorem{defi}[theo]{Definition}

\newtheorem{rem}[theo]{Remark}

\numberwithin{equation}{section}

\newcommand\hp[1]{\hphantom{#1}}

\def\bR{\mathbb{R}}
\def\bS{\mathbb{S}}
\def\bC{\mathbb{C}}
\def\mS{\mathcal{S}}
\def\div{\mathrm{div}}
\def\tr{\mathrm{tr}}
\def\Riem{\mathrm{Riem}}
\def\Ric{\mathrm{Ric}}
\def\d{\mathrm{d}}
\def\C{\mathrm{C}^\infty}
\def\Cc{\C_c}

\def\Ctc{\C_{tc}}
\def\f{\Omega}
\def\fc{\f_c}
\def\fsc{\f_{sc}}
\def\ftc{\f_{tc}}

\def\fcd{\f_{c\,\d}}

\def\h{\mathrm{H}}
\def\hc{\h_c}

\def\htc{\h_{tc}}
\def\s{\Gamma}
\def\sc{\s_c}
\def\ssc{\s_{sc}}
\def\stc{\s_{tc}}
\def\la{\langle}
\def\ra{\rangle}
\def\ker{\mathrm{Ker}}
\def\kersc{\ker_{sc}}
\def\kertc{\ker_{tc}}
\def\kerc{\ker_c}
\def\im{\mathrm{Im}}
\def\imsc{\im_{sc}}
\def\imtc{\im_{tc}}
\def\imc{\im_c}
\def\sol{\mathcal{S}}
\def\solsc{\sol_{sc}}
\def\gauge{\mathcal{G}}
\def\gaugesc{\gauge_{sc}}

\def\clobs{\mathcal{E}}
\def\obsk{\clobs^{kin}}
\def\obsi{\clobs^{inv}}
\def\radobs{\clobs^{rad}}
\def\obs{\mathcal{O}}

\title{%
Radiative observables for linearized gravity on asymptotically flat spacetimes and their boundary induced states
}

\author{%
Marco Benini$^{1,a}$, Claudio Dappiaggi$^{1,b}$ and Simone Murro$^{1,2,c}$\vspace{2mm}\\
{\small $^1$ Dipartimento di Fisica}\\ 
{\small Universit{\`a} di Pavia \& INFN, sezione di Pavia -- Via Bassi 6, I-27100 Pavia, Italia.}
\vspace{2mm}\\
{\small $^2$ Fakult\"at f\"ur Mathematik,}\\ 
{\small Universit\"at Regensburg, D-93040 Regensburg, Germany}
\vspace{4mm}
\\
 {\small  ~$^a$ marco.benini@pv.infn.it~,~$^b$ claudio.dappiaggi@unipv.it~,~$^c$ Simone.Murro@mathematik.uni-regensburg.de}
 }

\date{\today}

\begin{document}

\maketitle

\begin{abstract}
We discuss the quantization of linearized gravity on globally hyperbolic, asymptotically flat, vacuum spacetimes and the construction of distinguished states which are both of Hadamard form and invariant under the action of all bulk isometries. The procedure, we follow, consists of looking for a realization of the observables of the theory as a sub-algebra of an auxiliary, non-dynamical algebra constructed on future null infinity $\Im^+$. The applicability of this scheme is tantamount to proving that a solution of the equations of motion for linearized gravity can be extended smoothly to $\Im^+$. This has been claimed to be possible provided that a suitable gauge fixing condition, first written by Geroch and Xanthopoulos, is imposed. We review its definition critically showing that there exists a previously unnoticed obstruction in its implementation leading us to introducing the concept of radiative observables. These constitute an algebra for which a Hadamard state induced from null infinity and invariant under the action of all spacetime isometries exists and it is explicitly constructed.
\end{abstract}


\section{\label{sec:intro}Introduction}
The quantization of general relativity is one of the most debated, hard and treacherous topics in theoretical and mathematical physics. Much has been written on this subject, several models have been proposed and yet no unanimous solution has been found. A general consensus has been apparently reached in the form that, whatever is the correct answer, certainly it has either to account for non-perturbative effects or to realize Einstein's theory as the low energy regime of a more fundamental theory. The standard perturbation scheme, which has been successfully applied to many other cases, for instance quantum electrodynamics, has been thoroughly investigated and it is known since the eighties to be doomed to failure \cite{'tHooft:1974bx, Goroff:1985sz, Goroff:1985th} on account of the non renormalizability which becomes manifest at two loops.

From then on, the linearization of Einstein's equations has been often seen only in combination with an analysis of classical phenomena such as gravitational waves \cite[Section 4.4]{Wald} and it is hardly advocated to play any deep foundational role at a quantum level. In the same years, when new and mostly non perturbative approaches to quantum gravity were developed, several leaps forwards have been obtained also in our understanding of how to formulate rigorously free and perturbative quantum field theories on arbitrary backgrounds. The leading approach nowadays in this endeavour is known as algebraic quantum field theory, a framework which emphasizes that quantization should best be seen as a two-step procedure: In the first one assigns to a classical dynamical system a suitable $*$-algebra of observables $\mathcal{A}$, which encodes the mutual relational properties, locality, dynamics and causality in particular. In the second one constructs a state, namely a positive linear functional $\omega:\mathcal{A}\to\mathbb{C}$, from which one recovers the standard probabilistic interpretation via the GNS theorem -- see \cite{Benini:2013fia, Hollands:2014eia} for two recent reviews on the algebraic approach. 

If one is interested in free field theories, the Klein-Gordon or the Dirac field for instance, the first step is fully understood, insofar the underlying background is assumed to be globally hyperbolic: The dynamics of the system can be studied in terms of an initial value problem and the collection of all smooth solutions, {\it i.e.} configurations, can be constructed. Basic observables are then introduced via smooth and compactly supported sections of a suitable vector bundle which is dually paired to the collection of all configurations. With this procedure one associates unambiguously to the whole system a $*$-algebra of fields. More problematic is instead the identification of a state since, in the plethora of all possible choices, not all can be deemed to be physically acceptable. While on Minkowski spacetime, such quandary is bypassed simply by exploiting the covariance of the theory under the action of the Poincar\'e group so to single out a unique vacuum state, on curved backgrounds the situation is more complicated. In this case it is widely accepted that a state can be called physical if and only if it fulfils the Hadamard property, a condition on the singular structure of the (truncated) two-point function \cite{Kay:1988mu,Radzikowski:1996pa, Radzikowski:1996ei}. Such condition guarantees, on the one hand, that the ultraviolet behaviour of all correlation functions mimics that of the Poincar\'e vacuum and, on the other hand, that the quantum fluctuations of all observables are bounded. The importance of the Hadamard condition has been recently vigorously reaffirmed in \cite{Fewster:2013lqa}. Furthermore Hadamard states are of capital relevance in treating on curved backgrounds interactions at a perturbative level since they allow for an extension of the algebra of fields to encompass also Wick polynomials \cite{Hollands:2001nf}.

Despite being structurally so important, Hadamard states are known to be rather elusive to find unless the spacetime is static. Their existence is nonetheless guaranteed in most of the cases thanks to a deformation argument \cite{Fulling:1981cf}. Although such result is certainly important, for practical applications a constructive scheme is needed. Unless one considers very specific backgrounds, such as the cosmological ones, few options are known. The first aims at working directly on the initial value surface for the underlying equation of motion and it relies heavily on techniques proper of pseudo-differential calculus \cite{Gerard:2012wb, Gerard:2014wb}. The second is the one we will pursue and it is based on a procedure also dubbed {\em bulk-to-boundary correspondence}. It is a scheme which identifies for the underlying spacetime a distinguished codimension 1, null submanifold, such as for example the conformal boundary. Thereon one defines an auxiliary $*$-algebra as well as a distinguished quasi-free state. Subsequently, via a suitable homomorphism, it suffices to realize the algebra of fields as a sub-algebra of the auxiliary counterpart so to induce on the former, via pull-back, another state which turns out to be both Hadamard and invariant under all isometries. Such scheme has been applied successfully to several cases, ranging from the rigorous definition of the Unruh state for the wave equation on Schwarzschild spacetime \cite{Dappiaggi:2009fx}, to the identification of distinguished local states of Hadamard form \cite{Dappiaggi:2010iq}, to the construction of Hadamard states on asymptotically flat spacetimes for the free scalar \cite{Dappiaggi:2005ci}, Dirac \cite{Dappiaggi:2010gt} and Maxwell field \cite{Dappiaggi:2011cj, Siemssen:2011gma} -- for the latter there exists also another approach recently discussed in \cite{Finster:2013fva}. Conceptually all these approaches follow the same path proposed for the first time in \cite{Hollands}.

It is noteworthy that, until recently, all the investigations mentioned above never involved linearized gravity. Despite being a linear theory, it was not considered in algebraic quantum field theory. On the one hand one of the key problem is local gauge invariance. It is often thought to be difficult to reconcile with the algebraic approach and even the free Maxwell field suffered almost the same fate for this reason. On the other hand, the general philosophy according to which almost no insight on quantum gravity can be earned from perturbation theory has discouraged working on this topic. Yet, it appears that recently this trend has inverted. Particularly relevant is \cite{Brunetti:2013maa} in which not only perturbative quantum gravity is constructed in a generally covariant way, but a new light has been shed on the concept of observables for quantum gravity. In this context linearized gravity is seen as an important step to extract information about the local geometry of the full non-linear phase space. Most interestingly many results in \cite{Brunetti:2013maa} are actually derived using Hadamard states for linearized gravity, thus prompting the question of their explicit construction. The relevance of this question is increased by the realization that the standard deformation arguments, used for all other free fields adapting the analysis of \cite{Fulling:1981cf}, cannot be applied in this context since one is constrained to working with spacetimes which are solutions of the Einstein vacuum equations. 

Goal of this paper is thus to investigate an alternative mean to construct Hadamard states for linearized gravity. More precisely we will consider asymptotically flat vacuum spacetimes and we will investigate whether the bulk-to-boundary procedure, used successfully for all other free fields, can be implemented also in this case -- see also \cite{Frolov:1979ab} for an earlier related paper. In our analysis we will benefit mostly from a very recent and thorough analysis  on the construction of the algebra of fields for linearized gravity on an arbitrary globally hyperbolic background and on the definition thereon of Hadamard states \cite{Fewster:2012bj, Hunt2012}. The outcome of our investigation is rather surprising and it turns out to have far reaching consequences also for classical general relativity. As a matter of fact, the key point in the whole procedure is the following: Every asymptotically flat vacuum spacetime $(M,g)$ can be realized via an embedding $\psi$ as an open subset of a second auxiliary spacetime $(\widetilde M,\widetilde g)$ so that, in between other properties, the metrics $g$ and $\widetilde g$ are related by a conformal rescaling $\Xi$ on $\psi(M)$ and the boundary of $\psi(M)$ contains a null hypersurface, known as future (or past) null infinity $\Im^+$ ($\Im^-$). This can also be seen as the locus $\Xi=0$. On account of this geometric construction, we will show that realizing the algebra of observables in the bulk as a sub-algebra of a counterpart living on $\Im^+$, entails in particular proving that every gauge equivalence class of spacelike compact solutions of linearized Einstein's equations contains a representative which, up to a suitable conformal rescaling via $\Xi$, obeys in $\widetilde M$ to a hyperbolic partial differential equation. While in the scenarios considered previously in the literature, this feature was a direct consequence of the conformal invariance of the dynamics in the physical spacetime, linearized gravity behaves differently. More precisely, upon a conformal transformation, the equation of motion acquires terms which are proportional to inverse powers of $\Xi$ leading, thus, to a blow-up of the coefficients on $\Im^+$. This pathology can be in principle avoided exploiting the invariance of the theory under the action of the linearization of the diffeomorphism group of $(M,g)$ and finding a suitable gauge fixing which cancels the unwanted contributions. A positive answer to this question was found in the late seventies in \cite{Geroch:1978ur} and it played a key role both in establishing the stability of the notion of asymptotic flatness under linear perturbation and in studying the symplectic space of general relativity \cite{Ashtekar:1982aa}. 

We shall review in detail this construction and we will unveil that, in general, there exists an obstruction in implementing the so-called Geroch-Xanthopoulos gauge which depends both on the geometry and on the topology of the underlying background. As a matter of fact, as an example, we will show that Minkowski spacetime does not suffer from this problem while axisymmetric vacuum, asymptotically flat spacetimes do. In combination with the obstruction discovered by Fewester and Hunt in \cite{Fewster:2012bj, Hunt2012} to implement the transverse-traceless gauge, our result suggests that linearized gravity might be very much akin to electromagnetism. As a matter of fact, as shown in a series of papers \cite{Benini:2013ita, Benini:2013tra, DL, SDH12}, also the latter is affected by topological obstructions although these manifest explicitly as a consequence of Gauss' law when one tried to realize Maxwell equations as a locally covariant field theory. Although we will not investigate this specific issue, it might happen that also linearized gravity suffers of the same problem. 

As far as the construction of Hadamard states is concerned, our result suggests the introduction of the concept of radiative observables to indicate those which admit a counterpart on null infinity. In this way we identify a, not necessarily proper, sub-algebra of the full algebra of observables since, depending on the underlying spacetime, it can also coincide with the whole algebra of observables. Yet, for radiative observables there exists a bulk-to-boundary correspondence and we can thus identify via $\Im^+$ a distinguished state which is of Hadamard form and invariant under all bulk isometries.

The paper is organized as follows: In Section \ref{secLinGrav} we discuss linearized gravity on an arbitrary globally hyperbolic spacetime $(M,g)$ which solves the vacuum Einstein's equations. In particular we construct the space of gauge equivalence classes of solutions via the de Donder gauge fixing and we remark on the obstructions related to the transverse-traceless gauge. In Section \ref{subclobs} we define the classical observables as suitable equivalence classes of compactly supported, smooth symmetric $(2,0)$ tensors of vanishing divergence, endowing this space with a presymplectic structure. Section \ref{sec3} deals with the bulk-to-boundary correspondence. In Section \ref{sec3.1}, first we review the notion of an asymptotically flat spacetime and we outline the main geometric and structural properties of $\Im^+$, the conformal boundary. Afterwards we construct on $\Im^+$ a suitable symplectic space of smooth $(0,2)$ tensors. In Section \ref{sec3.2a} we discuss the Geroch-Xanthopoulos gauge and we prove the existence of an obstruction in its implementation. Examples are given and the concept of a radiative observable is introduced. In Section \ref{sec3.2} we show that there exists a symplectomorphism between the bulk radiative observables and the symplectic space on $\Im^+$ constructed in Section \ref{sec3.1}. In Section \ref{sec3.3} we extend this correspondence at a level of $*$-algebras and we exploit it to construct a state for radiative observables which is both of Hadamard form and invariant under the action of all isometries of the bulk metric. In Section \ref{sec5} we draw our conclusions.


\section{\label{secLinGrav}Linearized Gravity}

Aim of this section is to introduce the linearized version of Einstein's equations and to discuss 
how to assign a gauge invariant algebra of observables. Here we shall mimic the point of view 
used in \cite{Benini:2013tra} to tackle the same problem for Abelian principal connections. 
Furthermore, in our discussion we shall make use of the analyses both of Fewster and Hunt 
\cite{Fewster:2012bj, Hunt2012} concerning linearized gravity 
and of Hack and Schenkel concerning linear gauge theories \cite{Hack:2012dm}.

\subsection{Classical dynamics}\label{subClasDyn}

As a starting point we introduce the basic geometric ingredients, we shall use in this paper. 
We call {\bf spacetime} the collection $(M,g, \mathfrak{o}, \mathfrak{t})$, 
where $M$ is a four dimensional, smooth, connected  manifold, 
endowed with a smooth Lorentzian metric $g$ of signature $(-,+,+,+)$, 
as well as with the choice of an orientation $\mathfrak{o}$ and a time orientation $\mathfrak{t}$. 
We require additionally both that the metric solves vacuum Einstein's equations, {\em i.e.} $\Ric(g)=0$, 
and that $(M,g)$ is globally hyperbolic. In other words, there exists an achronal closed subset 
$\Sigma\subset M$ whose domain of dependence coincides with $M$ itself. 
This class of backgrounds is distinguished since it allows to discuss the dynamics 
of all standard free fields in terms of an initial value problem \cite{Bar:2007zz}.

We recall, moreover, that a subset $\Omega$ of a globally hyperbolic spacetime $(M,g)$ is said to be 
{\em future compact} ({\em resp. past compact}) if there exists a Cauchy surface $\Sigma\subset M$ 
such that $\Omega\subset J^-(\Sigma)$ ({\em resp.} $J^+(\Sigma)$), 
where $J^\pm$ stands for the causal future/past. 
At the same time $\Omega$ is called {\bf timelike compact} if it is both future and past compact, whereas we call it {\bf spacelike compact} if the support of the intersection of $\Omega$ with any Cauchy surface is either empty or compact. In this paper we shall employ the subscripts {\em tc} and {\em sc} to indicate maps whose support is timelike and spacelike compact respectively.

For later convenience, we will denote the bundle of symmetric tensors of order $n$ 
built out of a vector bundle $V$ with $S^nV$. We shall also need to work with  $\f^k(M)$ ($\fc^k(M)$) the space of smooth (and compactly supported) $k$-differential forms on which we define the exterior derivative $\d:\f^k(M)\to\f^{k+1}(M)$ and the codifferential $\delta:\f^k(M)\to\f^{k-1}(M)$ where $\delta\doteq (-1)^k*^{-1}\d*$, $*$ being the metric induced Hodge dual operator. 
Notice that, on Ricci flat backgrounds, the so-called Laplace-de Rham wave operator $\d\delta+\delta\d$ coincides with the metric induced wave operator $\Box$ if $k=0,1$. We shall indicate with $\f^k_\d(M)$ and $\f^k_\delta(M)$ the spaces of smooth forms which are respectively closed, {\it i.e.} $\d\omega=0$ and coclosed, {\it i.e.} $\delta\omega=0$ with $\omega\in\f^k(M)$. The same definitions are taken for those spaces of forms in which a suitable restriction on the support is assumed, {\it e.g.} compact, spacelike compact or timelike compact. We will also consider $\h^k(M)$ ($\hc^k(M)$), the $k$-th (compactly supported) de Rham cohomology group on $M$.

In order to discuss linearized gravity, we shall follow the scheme outlined in \cite{Stewart:1974uz} 
and recently analyzed from the point of view of algebraic quantum field theory in \cite{Fewster:2012bj, Hunt2012}. 
The dynamical variable is a smooth symmetric tensor field of type $(0,2)$, 
that is $h\in\s(S^2T^\ast M)$ 
which fulfils the so-called {\em linearized Einstein's equations}:
\begin{equation}\label{eqLinGrav}
-\frac{1}{2}g_{ab}\left(\nabla^c\nabla^d h_{cd}-\Box h\right)-
\Box\frac{h_{ab}}{2}-\frac{1}{2}\nabla_a\nabla_b h+\nabla^c\nabla_{(a}h_{b)c}=0,
\end{equation}
where $\nabla$ stands for the Levi-Civita connection for the metric $g$, 
$\Box\doteq g^{ab}\nabla_a\nabla_b$ and $h\doteq g^{ab}h_{ab}$. All indices are raised with respect to $g$. 
We also employ the standard symmetrization notation 
according to which $\nabla_{(a}h_{b)c}=\frac{1}{2}\left(\nabla_a h_{bc}+\nabla_b h_{ac}\right)$.

As one can imagine already by looking at \eqref{eqLinGrav}, it is advisable to employ a notation where indices are not spelled out explicitly, so to avoid unreadable formulas. We will try to adhere to this point of view as much as possible, although from time to time we will be forced to restore indices to make certain concepts more clear to a reader.

\begin{rem}
Eq. \eqref{eqLinGrav} is slightly simpler than in other papers. 
As a matter of fact it would suffice to require that $(M,g)$ were an Einstein manifold. 
This perspective is assumed for instance in \cite{Fewster:2012bj, Hunt2012}, 
namely the metric $g$ fulfils the cosmological vacuum Einstein's equation $G_{ab}+\Lambda g_{ab}=0$. 
In this setting additional terms, which are proportional to $\Lambda$, appear in \eqref{eqLinGrav}. 
In our case we are interested in constructing Hadamard states for asymptotically flat spacetimes, 
all fulfilling the vacuum Einstein's equations. For this reason we set from the very beginning $\Lambda=0$.
\end{rem}

\noindent Let us rewrite \eqref{eqLinGrav} in a more compact form. To start, we introduce four relevant operators:
\begin{itemize}
\item The {\em trace} $\tr:\s(S^2T^\ast M)\to\C(M)$ with respect to the background metric $g$ defined by
$$\tr h=g^{ab}h_{ab},$$
\item The {\em trace reversal} $I:\s(S^2T^\ast M)\to\s(S^2T^\ast M)$, which in components reads
$$(Ih)_{ab}\doteq h_{ab}-\frac{1}{2}g_{ab}h^c_{\hp{a}c},$$
\item The symmetrized covariant derivative, also known as the {\em Killing operator}, namely $\nabla_S:\s(S^nT^\ast M)\to\s(S^{n+1}T^\ast M)$ defined out of the Levi-Civita connection as 
\begin{equation*}
\left(\nabla_S H\right)_{i_0\dots i_n}\doteq \nabla_{(i_0}H_{i_1\dots i_n)},
\end{equation*}
where the brackets stand for the normalized symmetrization of the indices. 
\item The {\em divergence operator} $\div:\s(S^{n+1}T^\ast M)\to\s(S^nT^\ast M)$ such that
\begin{equation*}
(\div H)_{i_0\dots i_n}\doteq g^{ab}\nabla_a H_{bi_1\dots i_n}.
\end{equation*}
\end{itemize} 
Similar definitions are taken for sections of $TM$ and of its symmetric tensor powers. Let us consider the standard pairing between sections of $S^nT^\ast M$ and $S^nT M$: 
\begin{equation}\label{eqPairing}
(\cdot,\cdot):\s(S^nT^\ast M)\check{\times}\s(S^nTM)\to\bR,\quad(H,\Xi)=\int\limits_M\la H,\Xi\ra\,\mu_g,
\end{equation}
where $\la\cdot,\cdot\ra$ is the dual pairing between $S^nT^\ast M$ and $S^nTM$ 
and $\check{\times}$ denotes the subset of the Cartesian product whose elements are pairs with compact overlapping support. Accordingly $\nabla_S$ and $-\div$ are dual to each other, namely for each pair $(H,A)\in \s(S^{n+1}T^\ast M)\check{\times}\s(S^nTM)$:
\begin{equation}\label{eqEradDiv}
(H,\nabla_S A)=(-\div H,A).
\end{equation}
Since the background $(M,g)$ is Ricci flat per assumption, the Riemann tensor $R_{abcd}$ is the only non vanishing geometric quantity which plays a distinguished role. We shall employ it to define an operation on symmetric $(0,2)$-tensors as follows:
\begin{equation*}
\Riem:\s(S^2T^\ast M)\to\s(S^2T^\ast M),\quad \Riem(h)_{ab}=R_{a\hp{cd}b}^{\hp{a}cd} h_{cd}.
\end{equation*}
The counterpart for $(2,0)$-tensors follows suit. Although it is almost unanimously accepted that indices of tensors are raised and lowered via the metric, it is useful to introduce explicitly the musical isomorphisms which indicate such operations. We shall employ them in those formulas where their action is better spelled out in order to avoid possible confusion:
\begin{align*}
\cdot^\sharp: T^\ast M^{\otimes n}\to TM^{\otimes n},
& \quad(H^\sharp)^{i_1\dots i_n}= g^{i_1j_1}\cdots g^{i_nj_n}H_{j_1\dots j_n}\\
\cdot^\flat: TM^{\otimes n}\to T^\ast M^{\otimes n},
& \quad(A^\flat)_{i_1\dots i_n}=g_{i_1j_1}\cdots g_{i_nj_n}A^{j_1\dots j_n}
\end{align*}

Up to an irrelevant factor $\frac{1}{2}$, we can rewrite \eqref{eqLinGrav} in terms of the operators introduced above:
\begin{equation}\label{eqLinGravOp}
K:\s(S^2T^\ast M)\to\s(S^2T^\ast M),\quad K=(-\Box+2\Riem+2I\nabla_S\div)I.
\end{equation}
By direct inspection one can infer that $K\circ I$ is not a normally hyperbolic operator due to the presence of the term $I\nabla_S\div$ -- for a detailed account of the definition and of the properties of these operators, refer to the monograph \cite{Bar:2007zz}. In other words, even on globally hyperbolic spacetimes, the solutions of \eqref{eqLinGravOp} cannot be constructed via a Cauchy problem with arbitrary initial data. On the contrary, suitable initial data must satisfy appropriate constraints, see {\it e.g.} \cite[Section 3.1]{Fewster:2012bj} and \cite[Section 4.6]{Hunt2012}. Especially from the point of view of quantizing the theory, this feature is rather problematic. Yet, as we will discuss in detail in the next section, we can circumvent this issue exploiting the underlying gauge invariance of \eqref{eqLinGrav}. Since it plays a distinguished role in our calculations, we make explicit how \eqref{eqLinGravOp} intertwines with the geometric operators introduced above. The proof of the following lemma is a matter of using the standard structural properties of the covariant derivative and it will be therefore omitted.
\begin{lem}\label{lemma_useful}
The following tensorial identities hold true:
\begin{itemize}
\item[1.] $(\Box-2\Riem)\nabla_S=\nabla_S\Box$ on $\s(TM)$;
\item[2.] $2\div I\nabla_S=\Box$ on $\s(TM)$;
\item[3.] $(\Box-2\Riem)I=I(\Box-2\Riem)$ on $\s(S^2TM)$;
\item[4.] $\tr(\Box-2\Riem)=\Box\tr$ on $\s(S^2TM)$.
\end{itemize}
Similar identities hold true for $T^\ast M$ via musical isomorphisms. 
Moreover one can derive dual identities exploiting the pairing \eqref{eqPairing}. 
\end{lem}

\subsection{Gauge symmetry and gauge fixed dynamics}\label{subGaugeDyn}

The well-known diffeomorphism invariance of general relativity translates at the level of the linearized theory in the following gauge freedom:
Any solution $h$ of \eqref{eqLinGrav} is equivalent to $h+\nabla_S\chi$, for arbitrary $\chi\in\s(T^\ast M)$:
\begin{equation}\label{eqGauge}
h\sim h^\prime\in\s(S^2T^\ast M)\quad\iff\quad\exists\chi\in\s(T^\ast M):\,h^\prime=h+\nabla_S\chi.
\end{equation}
Notice that, consistently with gauge invariance, for all $\chi\in\s(T^\ast M)$, $K\nabla_S\chi=0$ as it can be readily verified via Lemma \ref{lemma_useful}. Following the same line of reasoning as for other linear gauge theories \cite{Hack:2012dm}, we can turn eq. \eqref{eqLinGrav} to an hyperbolic one by a suitable gauge fixing. For linearized gravity the first choice is usually the so-called {\bf de Donder gauge}, which plays the same role of the Lorenz gauge in electromagnetism. Although this is an overkilled topic, we summarize the main results in the following proposition \cite{Fewster:2012bj, Hunt2012}:

\begin{propo}\label{prpDDEauge}
Let $\sol/\gauge$ be the space of gauge equivalence classes $[h]$ of solutions of \eqref{eqLinGrav}, 
that is $\sol\doteq\{h\in\s(S^2T^\ast M)\;|\;Kh=0\}$, where $K$ has the form \eqref{eqLinGravOp}, while 
$\gauge\doteq\{h\in\s(S^2T^\ast M)\;|\;h=\nabla_S\chi,\;\chi\in\s(T^*M)\}$. 
Then, for every $[h]\in\sol$, there exists a representative $\widetilde{h}\in[h]$ satisfying
\begin{subequations}\label{eqLinGravEF}
\begin{align}
P\widetilde{h}=(\Box-2\Riem)I\widetilde{h} & = 0,\label{eqLinGravDD}\\
\div I\widetilde{h} & = 0.\label{eqDDGauge}
\end{align}
\end{subequations}
\end{propo}
\begin{proof}
Let $[h]\in\sol$ and let us choose any representative $h$. Suppose $h$ does not satisfy \eqref{eqDDGauge}; we look therefore for another representative $\widetilde{h}\in [h]$ such that $\div I\widetilde{h}=0$. To this end, consider $\chi\in\s(T^*M)$, solution of $\Box\chi+2\div Ih=0$. The existence of such $\chi$ descends from applying \cite[Section 3.5.3, Corollary 5]{Bar:2009zzb} since $\Box$ is normally hyperbolic and $\div Ih$ is a smooth source term.
Let $\widetilde{h}=h+\nabla_S\chi$. Since it is per construction a solution \eqref{eqLinGravOp}, which satisfies furthermore \eqref{eqDDGauge}, 
also \eqref{eqLinGravDD} holds true, as one can verify directly.
\end{proof}

We can use the last proposition to characterize the gauge equivalence classes of solutions of \eqref{eqLinGrav} via \eqref{eqLinGravEF}. 
As a starting point, notice that the operator $\widetilde P\doteq\Box-2\Riem$ acting on $\s(S^2T^{\ast} M)$ is normally hyperbolic. As explained in details in \cite{Bar:2013,Sanders:2012ac}, this is tantamount to the existence of unique retarded and advanced Green operators $\widetilde E^\pm:\stc(S^2T^{\ast} M)\to\s(S^2T^{\ast} M)$ such that $\widetilde P\widetilde E^\pm =\mathrm{id}$ and $\widetilde E^\pm\widetilde P=\mathrm{id}$, where $\mathrm{id}$ stands for the identity operator on $\stc(S^2T^{\ast} M)$. The following support properties hold true: $\textrm{supp}(\widetilde E^\pm(\epsilon))\subseteq J^\pm(\textrm{supp}(\epsilon))$ for any $\epsilon\in\stc(S^2T^{\ast} M)$.

Since $P=\widetilde P\circ I$, where $I$ is the trace-reversal operator, item 3. of Lemma \ref{lemma_useful} guarantees us that, although $P$ is not normally hyperbolic, it is Green hyperbolic, that is one can associate to it retarded and advanced Green operators $E^\pm\doteq I\circ\widetilde E^\pm = \widetilde E^\pm\circ I$, which share the same support properties of $\widetilde E^\pm$ and are both right and left inverses of $P$. Actually, $E^\pm$ are also unique, $P$ being formally self-adjoint, {\it i.e.} $(Ph,k^\sharp)=(h,(Pk)^\sharp)$ for each $h,k\in\s(S^2T^\ast M)$ with compact overlapping support. Let $E=E^+-E^-$ be the so-called {\em causal propagator}. On account of the exactness of the sequence
$$0\longrightarrow\stc(S^2T^\ast M)\overset{P}{\longrightarrow}\stc(S^2T^\ast M)
\overset{E}{\longrightarrow}\s(S^2T^\ast M)\overset{P}{\longrightarrow}\s(S^2T^\ast M)\longrightarrow 0,$$
the causal propagator $E$ induces an isomorphism between $\stc(S^2T^\ast M)/\imtc P$ and the space of solutions of \eqref{eqLinGravDD} -- see \cite{Khavkine:2014kya}. 

Yet, since we are looking for gauge equivalence classes of solutions of \eqref{eqLinGrav}, we are interested only in those solutions of \eqref{eqLinGravDD} fulfilling also \eqref{eqDDGauge}. To select them, we follow the same strategy as first outlined in \cite{Dimock:1992ff} for the free Maxwell equations, namely we translate the gauge fixing condition in the restriction to a suitable subspace of $\stc(S^2T^{(\ast)} M)$. As a preliminary result we prove the following lemma:

\begin{lem}\label{lemUseful}
Let $E^\pm$ be the retarded (+) and advanced (-) Green operators of $P$. Then on $\stc(S^2T^\ast M)$ it holds
$$\div IE^\pm=E_\Box^\pm\div,$$
where $E_\Box^\pm$ denotes the retarded/advanced Green operator for $\Box$ acting on sections of $T^\ast M$.
\end{lem}
\begin{proof}
Via the dual pairing between $TM$ and $T^\ast M$, we deduce that $(E_\Box^\pm\xi,\chi)=(\xi,E_\Box^\mp\chi)$ 
for each $\xi\in\stc(T^\ast M)$ and each $\chi\in\sc(TM)$. Notice that, with a slight abuse of notation, we use the same symbol $E_\Box^\pm$ for the Green operators of $\Box$ when it acts on sections of both $TM$ and $T^\ast M$.  We can now refer to item 1. in Lemma \ref{lemma_useful} to conclude $\div IE^\pm=E_\Box^\pm\div$ on $\stc(S^2T^\ast M)$. As a matter of fact, for each $v\in\stc(S^2T^\ast M)$ and $\xi\in\sc(TM)$ it holds
\begin{align*}
(\div IE^\pm v,\xi) & =-(E^\pm v,I\nabla_S\xi)=-(E^\pm v,I\nabla_S\Box E_\Box^\mp\xi)
=-(E^\pm v,I(\Box-2\Riem)\nabla_S E_\Box^\mp\xi)\\
& =-((\Box-2\Riem)IE^\pm v,\nabla_S E_\Box^\mp\xi)=-(v,\nabla_S E_\Box^\mp\xi)=(E_\Box^\pm\div v,\xi).
\end{align*}
The arbitrariness of $\xi$ entails the sought result.
\end{proof}

\begin{rem}\label{remResEauge}
Notice that the de Donder gauge fixing is not complete, namely there is a residual gauge freedom. 
In other words, for each $[h]\in\sol$, there exists more than one representative fulfilling both \eqref{eqLinGravDD} and \eqref{eqDDGauge}. These representatives differ by pure gauge solutions of the form $\nabla_S\chi$, where $\chi\in\s(T^\ast M)$ is such that $\Box\chi=0$.
\end{rem}

We characterize the solutions of the equations of motion for linearized gravity in the de Donder gauge:

\begin{propo}\label{prpEFSol}
Let $\kertc(\div)\doteq\{\epsilon\in\stc(S^2T^\ast M)\;|\;\div\epsilon=0\}$. Then $E\epsilon$ solves both \eqref{eqLinGravDD} and \eqref{eqDDGauge}, where $E$ is the causal propagator of $P$. Conversely, for each solution $h\in\s(S^2T^\ast M)$ of the system \eqref{eqLinGravEF}, there exists $\epsilon\in\kertc(\div)$ such that $E\epsilon$ differs from $h$ by a residual gauge transformation according to Remark \ref{remResEauge}.
\end{propo}
\begin{proof}
Let $\epsilon$ be in $\kertc(\div)$. Per definition of $E$, we know that $E\epsilon$ solves \eqref{eqLinGravDD}. On account of Lemma \ref{lemUseful}, it holds also that $\div IE\epsilon=E_\Box\div\epsilon=0$; hence $E\epsilon$ solves also \eqref{eqDDGauge}. 

Consider now any solution $h$ of \eqref{eqLinGravEF}. As a consequence of \eqref{eqLinGravDD}, there exists $\widetilde\epsilon\in\stc(S^2T^\ast M)$ such that $E(\widetilde\epsilon)=h$. Still on account of Lemma \ref{lemUseful}, \eqref{eqDDGauge} translates into $E_\Box\div\widetilde\epsilon=0$. Hence there exists $\eta\in\stc(T^\ast M)$ such that $\div\widetilde\epsilon=\Box\eta$. Let $\epsilon\doteq \widetilde\epsilon-2I\nabla_S\eta\in\stc(S^2T^\ast M)$. On account of item 2. in Lemma \ref{lemma_useful}, it holds that $\div\epsilon=0$. Furthermore 
$$E\epsilon=E\widetilde\epsilon-2EI\nabla_S\eta=h-2\nabla_S E_\Box\eta.$$
Therefore $E\epsilon$ differs from $h$ by $\nabla_S\chi$, where $\chi\doteq  -2E_\Box\eta$ solves $\Box\chi=0$ and, thus, it is a residual gauge transformation.
\end{proof}

We are ready to characterize the space of gauge equivalence classes of solutions of linearized gravity:

\begin{theo}\label{thmSolModEauge}
There exists a one-to-one correspondence between $\sol/\gauge$, the set of gauge equivalence classes of solutions of \eqref{eqLinGrav}, and the quotient between $\kertc(\div)$ and $\imtc(K)\doteq\{\epsilon\in\stc(S^2T^\ast M)\;|\;\epsilon=K\gamma,\;\gamma\in\stc(S^2T^\ast M)\}$. The isomorphism is explicitly realized by the map $[\epsilon]\mapsto [E\epsilon]$ for $[\epsilon]\in\kertc(\div)/\imtc(K)$.
\end{theo}

\begin{proof}
As a starting point, notice that the quotient between $\kertc(\div)$ and $\imtc(K)$ is meaningful since the identity $\div\circ K=0$ holds on $\s(S^2T^\ast M)$. This follows from eq. \eqref{eqEradDiv}, $K\circ\nabla_S=0$ on $\s(S^2T^\ast M)$ and $K$ being formally self-adjoint, namely $(Kh,v^\sharp)=(h,(Kv)^\sharp)$ for each $h,v\in\s(S^2T^\ast M)$ with compact overlapping support. 

On account of Proposition \ref{prpEFSol} we know that $E$ associates to each element in $\kertc(\div)$ a solution of both \eqref{eqLinGravDD} and \eqref{eqDDGauge}. In turn this identifies a unique gauge equivalence class in $\sol/\gauge$. Since $E\circ P=0$ and taking into account the dual of Lemma \ref{lemUseful}, one reads $E\circ K=2EI\nabla_S\div I=2\nabla_SE^\pm_\Box\div I$, it holds also that the map $[\epsilon]\mapsto [E\epsilon]$ is well-defined. Furthermore, Proposition \ref{prpDDEauge} and Proposition \ref{prpEFSol} together entail that this application is surjective since, for any $[h]\in\sol/\gauge$, one has a representative $h$ in the de Donder gauge and there exists $\epsilon\in\kertc(\div)$ such that $E\epsilon$ differs from $h$ at most by a residual gauge transformation. In other words we have found $[\epsilon]\in\kertc(\div)/\imtc(K)$ such that $[E\epsilon]=[h]$.  

Let us now prove that the map induced by $E$ is injective. This is tantamount to prove the following statement: $[\epsilon]\in\kertc(\div)/\imtc(K)$ such that $[E\epsilon]=[0]\in\sol/\gauge$ entails $[\epsilon]=0$. Accordingly, let us assume that $\epsilon\in\kertc(\div)$ is such that $E\epsilon=\nabla_S\chi$ for $\chi\in\s(T^\ast M)$. By \eqref{eqDDGauge} and Lemma \ref{lemUseful}, this entails that $\Box\chi=0$ and hence $\chi=E_\Box\alpha$, for a suitable $\alpha\in\stc(T^\ast M)$. On account of the dual of Lemma \ref{lemUseful}, it holds $E\epsilon=EI\nabla_S\alpha$, that is there exists $\beta\in\stc(S^2T^\ast M)$ such that $\epsilon=I\nabla_S\alpha+P\beta$. Acting with $\div$ on both sides, one obtains $0=\Box\alpha+2\Box\div I\beta$, or in other words $\alpha=-2\div I\beta$. Hence, $\epsilon=P\beta-2I\nabla_S\div I \beta=-K\beta$. This concludes the proof.
\end{proof}

\begin{rem}\label{remTTTC}
In several discussions concerning linearized gravity, radiative degrees of freedom or even gravitational waves, 
it is customary to exploit the additional gauge freedom (see Remark \ref{remResEauge}) 
to switch from the de Donder to the so-called {\em transverse-traceless (TT)} gauge. 
This consists of adding one more constraint on the field configuration besides eq. \eqref{eqDDGauge}. 
The extra condition reads $\tr h=0$, 
where $\tr:\s(S^2 T^\ast M)\to\C(M)$ denotes the trace computed using the metric $g$, 
motivating the word {\em traceless}. 
It was noticed for the first time in \cite{Fewster:2012bj, Hunt2012} 
that this kind of gauge fixing is not always possible on account of a topological constraint. 
We discuss a complementary approach to this problem: 
Considering any solution $h$ of the system \eqref{eqLinGravEF}, 
we try to exploit the residual gauge freedom explained in Remark \ref{remResEauge} 
to find $\widetilde{h}\sim h$ such that $\widetilde{h}$ is also traceless. 
Thus we look for $\chi\in\s(T^*M)$ such that $\Box\chi=0$ and $\tr h+\div\chi=0$. 
Since $\chi$ is nothing but a $1$-form, $\chi\in\f^1(M)$, 
it is convenient to read these conditions using the exterior derivative $\d$ 
and the codifferential $\delta$ defined on $1$-forms. 
In other words we are looking for solutions of the following two equations: 
\begin{equation}\label{eqTTFix}
\Box\chi=0,\qquad\delta\chi=\tr h,
\end{equation}
where the D'Alembert wave operator $\Box$ coincides here 
with the Laplace-de Rham wave operator $\delta\d+\d\delta$ since the background is Ricci flat. 
Here we follow the nomenclature of \cite{SDH12} using the subscript $(p)$ attached to $E$ to indicate 
the causal propagator for the Laplace-de Rham wave operator acting on $p$-forms. 
Actually, it is convenient to exploit Proposition \ref{prpEFSol} 
to express the given $h$ in terms of a suitable $\epsilon\in\kertc(\div)$, namely $h=E\epsilon$. 
In particular, we are interested in the identity $\tr h=\tr E\epsilon =-E_{(0)}(\tr\epsilon)$, 
which is a consequence of item 4. in Lemma \ref{lemma_useful}. 

We show below that the system of equations \eqref{eqTTFix} admits a solution if and only if there exists 
$\lambda\in\ftc^1(M)$ such that $\tr\epsilon=\delta\lambda$. 
Assume first that there exists a solution $\chi$ of the system \eqref{eqTTFix}. 
From the first equation, $\chi$ is of the form $\chi=E_{(1)}\alpha$ for a suitable $\alpha\in\ftc^1(M)$. 
Therefore, the second equation entails $E_{(0)}\delta\alpha=-E_{(0)}(\tr\epsilon)$. 
In turn this fact entails $\delta\alpha=-\tr\epsilon+\Box f$ for $f\in \Ctc(M)$, 
that is to say $\tr\epsilon=\delta(\d f-\alpha)$. 
Conversely, suppose $\tr\epsilon=\delta\lambda$ for $\lambda\in\ftc^1(M)$. 
Then one can directly check that $\chi=-E_{(1)}\lambda$ is a solution of the system above. 

In conclusion, given a de Donder solution $h=E \epsilon$, $\epsilon\in\sc(S^2T^\ast M)$, 
it is possible to achieve the TT gauge 
if and only if $\ast(\tr\epsilon)\in\d\ftc^3(M)$, that is to say $[\ast(\tr\epsilon)]=0\in \htc^4(M)$, 
the fourth de Rham cohomology group with timelike compact support -- for a recent analysis see \cite{Benini, Igor}. 
Following \cite[Theorem 5.5]{Benini}, $\htc^4(M)\simeq \h^3(\Sigma)$, $\Sigma$ being a Cauchy surface for the globally hyperbolic spacetime $M$. Therefore, the TT gauge can be always achieved for solutions without any restriction on the support 
provided $\h^3(\Sigma)=0$, namely, via Poincar\'e duality, the Cauchy surface is non-compact. 
Otherwise there might be obstructions to the TT gauge for certain solutions. 

Notice that, in principle, restricting our attention only to those asymptotically flat and globally hyperbolic spacetimes having non-compact Cauchy surfaces might not be such a severe restriction since, actually, we are unaware of an explicit example falling outside the class just mentioned -- see for example \cite{Murro}. Yet, in the next sections, in order to characterize a well-defined bulk-to-boundary projection, we shall employ another gauge fixing due to Geroch and Xanthopoulos. This procedure displays a very similar obstruction, related both to the geometry and to the topology of the underlying background.
\end{rem}

\subsection{\label{subclobs}Classical observables for linearized gravity}

Our goal is to construct an algebra of observables 
which encompasses both the dynamics \eqref{eqLinGrav} and the gauge symmetry \eqref{eqGauge}. 
In this respect we employ a procedure partly different from the one in \cite{Fewster:2012bj, Hunt2012}. 
We mimic, instead, the approach of \cite{Benini:2012vi}, 
subsequently applied in \cite{Benini:2013tra} to the quantization of principal Abelian connections. 
In our opinion this approach has several net advantages. For instance, in the case of Yang-Mills theory with an Abelian gauge group containing at least a $U(1)$ factor, it helps unveiling the optimal class of observables to test the collection of gauge equivalent configurations, {\em cfr.} \cite{Benini:2013tra, Benini:2013ita}.  Also in the case, we consider, this dual (in a sense specified below) approach is still worth following since it avoids any a priori gauge fixing to construct the functionals which define observables. 
This feature is particularly relevant in clarifying the key aspects of the bulk-to-boundary correspondence 
for linearized gravity on asymptotically flat spacetimes, 
for which several choices of gauge fixings appear at different stages of the procedure.

Following \cite{Benini:2012vi, Benini:2013tra}, 
we start from the space of off-shell field configurations, namely $\s(S^2T^\ast M)$. 
Recalling \eqref{eqPairing}, a convenient space of sections which is dually paired to this one 
is given by compactly supported sections of the dual bundle, $\sc(S^2TM)$.
For any $\epsilon\in\sc(S^2TM)$, we can introduce a linear functional $\obs_\epsilon$ as follows:
\begin{equation*}
\obs_\epsilon:\s(S^2T^\ast M)\to\mathbb{R},\quad\obs_\epsilon(h)=(h,\epsilon),
\end{equation*}
where $\mu_g$ is the metric induced volume form on $M$. 
Notice that the evaluation of $\mathcal{O}_\epsilon$ on the off-shell configuration $h$ is nothing 
but the usual pairing between $h\in\s(S^2T^\ast M)$ and $\epsilon\in\sc(S^2TM)$. 
The collection of all functionals $\obs_\epsilon$, $\epsilon\in\sc(S^2TM)$, 
forms a vector space which we indicate as $\obsk$ 
and, due to non-degeneracy of the pairing $(\cdot,\cdot)$ introduced in \eqref{eqPairing}, 
it is isomorphic to $\sc(S^2TM)$. 
Hence we will often identify $\obsk$ with $\sc(S^2TM)$ 
by writing $\epsilon\in\obsk$ for any $\epsilon\in\sc(S^2TM)$.

Up to this point, a functional $\obs_\epsilon$, $\epsilon\in\sc(S^2TM)$, 
is neither invariant under gauge transformations nor on-shell, 
requirements which are both needed in order to interpret $\obs_\epsilon$ as an observable 
for the classical field theory describing linearized gravity. 
As a first step, we identify those functionals which behave properly under gauge transformations. 
Recalling \eqref{eqGauge}, we realize that $\obs_\epsilon\in\obsk$ is {\bf gauge invariant} 
if and only if $\obs_\epsilon(\nabla_S\chi)=0$ for all $\chi\in\s(T^\ast M)$. 
The following lemma characterizes gauge invariant functionals:

\begin{lem}\label{dualgauge}
A functional $\obs_\epsilon$ is invariant under gauge transformations if and only if $\div\epsilon=0$.
\end{lem}
\begin{proof}
This follows from eq. \eqref{eqEradDiv}, which states that $\nabla_S$ and $-\div$ are the dual of each other. 
This entails that, for all $\epsilon\in\sc(S^2TM)$ and for all $\chi\in\s(T^*M)$, 
\begin{equation*}
\obs_\epsilon(\nabla_S\chi)=(\nabla_S\chi,\epsilon)=(\chi,-\div\epsilon). 
\end{equation*}
Since the pairing between $\sc(TM)$ and $\s(T^*M)$ is non-degenerate, we deduce that 
$\obs_\epsilon$ is gauge invariant, namely $\obs_\epsilon(\nabla_S\chi)=0$ for each $\chi\in\s(T^\ast M)$ 
if and only if $\div\epsilon=0$. 
\end{proof}
\noindent Lemma \ref{dualgauge} motivates the definition given below 
for the space of gauge invariant linear functionals: 
\begin{equation*}
\obsi=\left\{\obs_\epsilon\in\obsk:\;\div\epsilon=0\right\}=\kerc(\div).
\end{equation*}

As a last step, we have to account for the dynamics. More precisely we wish to construct equivalence classes, identifying two elements in $\obsi$ whenever they differ by a third one which yields $0$ when evaluated on any configuration $h\in\s(S^2T^*M)$ solving \eqref{eqLinGrav}. This is achieved taking the quotient of $\obsi$ by the image of the dual of the differential operator $K$ 
ruling the dynamics. In this way we obtain (classes of) gauge invariant functionals 
whose evaluation is well-defined only on (gauge equivalence classes of) solutions to the field equation $Kh=0$. 
We proceed as follows: First, we compute the dual of $K$ with respect to 
the pairing $(\cdot,\cdot)$ between sections of $S^2T^\ast M$ and $S^2TM$ defined in \eqref{eqPairing}. 
For each $h\in\s(S^2T^\ast M)$ and each $\epsilon\in\s(S^2TM)$ with compact overlapping support, we have 
\begin{align*}
(h,K^\ast\epsilon) & =(Kh,\epsilon)=((-\Box+2\Riem+2I\nabla_S\div)Ih,\epsilon)
=(I(-\Box+2\Riem+2\nabla_S\div I)h,\epsilon)\\
& =(h,(-\Box+2\Riem+2I\nabla_S\div)I\epsilon)=(h,(K\epsilon^\flat)^\sharp),
\end{align*}
where we used item 3. in Lemma \ref{lemma_useful}, 
standard properties of the Riemann tensor as well as \eqref{eqEradDiv}. 
This entails that the dual $K^\ast$ of $K$ coincides with $K$ itself, 
barring the musical isomorphisms to pass from $TM$ to $T^\ast M$ and vice versa. 
For this reason, with a slight abuse of notation, in the following we will use $K$ also to denote its dual $K^\ast$. 
We introduce the space of classical observables as follows: 
\begin{equation}\label{obs}
\clobs=\frac{\obsi}{\imc(K)}.
\end{equation}
Notice that the quotient is well-defined on account of the identity $\div\circ K=0$ on $\s(S^2TM)$. 
The evaluation of an observable $[\epsilon]\in\clobs$ 
on a gauge class of solutions $[h]\in\sol/\gauge$ is consistently obtained 
by an arbitrary choice of representatives, namely $\obs_{[\epsilon]}([h])=\obs_\epsilon(h)$ 
for each $\epsilon\in[\epsilon]$ and $h\in[h]$. 
The space of classical observables $\clobs$ can be endowed with a presymplectic form. 
This is introduced via the causal propagator $E$ of the Green hyperbolic differential operator $P$, 
which rules the gauge-fixed dynamics, see Proposition \ref{prpDDEauge}. 

\begin{propo}\label{prpPresympl}
The space of classical observables $\clobs$ can be endowed with the presymplectic structure\footnote{From a geometric point of view, it would be more customary and appropriate to talk about a constant Poisson structure -- see for example \cite{Khavkine:2014kya}. We will stick to the nomenclature more commonly used in quantum field theory on curved backgrounds.} defined below:
\begin{equation}\label{tau}
\tau:\clobs\otimes\clobs\to\bR,\quad\tau([\epsilon],[\zeta])=2(E\epsilon^\flat,\zeta),
\end{equation}
where $E$ is the causal propagator for $P$ and where the right-hand-side is written in terms 
of an arbitrary choice of representatives in both equivalence classes.
\end{propo}
\begin{proof}
First, let us notice that $(E\,\cdot^\flat,\cdot)$ is bilinear and skew-symmetric on $\obsk$. 
Bilinearity can be directly read from the formula, 
hence we take arbitrary $\eta,\zeta\in\obsk$ and show that $(E\epsilon^\flat,\zeta)=-(E\zeta^\flat,\epsilon)$: 
$\cdot^\flat$ intertwines both $P$ and its dual (still denoted by $P$ with a slight abuse of notation 
motivated by the fact that the operators looks exactly the same up to musical isomorphisms). 
Therefore a similar property holds true for the corresponding causal propagators, 
both denoted by $E$ with the same abuse of notation. 
Since we are dealing with $P$ and its dual, the following relation between the corresponding Green operators holds: 
$(E^\pm h,\eta)=(h,E^\mp\eta)$ for each $h\in\sc(S^2T^\ast M)$ and $\eta\in\sc(S^2TM)$. 
To conclude this part of the proof, we stress that $(A^\flat,H^\sharp)=(H,A)$ 
for each $A\in\s(S^nTM)$ and $H\in\s(S^nT^\ast M)$ with compact overlapping support. 
All these considerations entail that 
\begin{equation*}
(E\epsilon^\flat,\zeta)=((E\epsilon)^\flat,\zeta)=(\zeta^\flat,(E\epsilon)^{\flat\,\sharp})
=(\zeta^\flat,E\epsilon)=-(E\zeta^\flat,\epsilon).
\end{equation*}
Up to this point, we have a presymplectic form $(E\,\cdot^\flat,\cdot)$ on $\obsk$, hence in particular on $\obsi$. 
We still have to prove that such structure descends to the quotient space $\clobs$, 
thus providing $\tau$ as specified in the statement: 
To this end, we take $\epsilon\in\obsi$ and $\eta\in\sc(S^2TM)$ and show that $(E\epsilon^\flat,K\eta)=0$. 
To proceed, we take into account that $(E\,\cdot^\flat,\cdot)$ is skew-symmetric, 
we recall the definition of $K$, eq. \eqref{eqLinGravOp}, and we consider its dual acting on sections of $S^2TM$, 
still denoted by $K$ with the usual abuse of notation. 
Exploiting item 1. of Lemma \ref{lemma_useful} and \eqref{eqEradDiv}, we end up with 
\begin{equation*}
(E\epsilon^\flat,K\eta)=-(EK\eta^\flat,\epsilon)=-(EI\nabla_S\div I\eta^\flat,\epsilon)
=-(\nabla_S E_\Box\div I\eta^\flat,\epsilon)=(E_\Box\div I\eta^\flat,\div\epsilon)=0,
\end{equation*}
where the last equality follows from gauge-invariance of $\epsilon$. 
Therefore $(E\,\cdot^\flat,\cdot)$ descends to the quotient $\clobs$. 
This shows that $\tau$ is a well-defined presymplectic form on $\clobs$, thus completing the proof. 
\end{proof}

\noindent The factor $2$ appearing in the expression of the presymplectic form might look unusual, as much as the fact that the causal propagator appears in the left slot of the pairing. Here we are using the causal propagator $E$ for $P=(\Box-2\Riem)I$ 
in order to define the presymplectic form, which does not take into account a factor $-1/2$ appearing in \eqref{eqLinGrav}.  
The minus sign is compensated indeed by the causal propagator on the left side. Notice that we made no statement about the non-degeneracy of \eqref{tau}. At the moment a positive answer has been given by Fewster and Hunt for globally hyperbolic spacetimes with compact Cauchy surfaces \cite{Fewster:2012bj} and recently by Hack on Minkowski spacetime \cite{Hack}. We will not dwell into this problem since it does not play a significant role in our investigation. 

\begin{rem}
It is possible to obtain the same formula for the presymplectic structure in Proposition \ref{prpPresympl} 
generalizing a method originally due to Peierls \cite{Peierls:1952cb} to gauge theories \cite{Khavkine:2012jf, Khavkine:2014kya}. 
This approach was considered already in \cite{Marolf:1993af} 
and was recently put on mathematically solid grounds in \cite{SDH12} for the vector potential of electromagnetism. 
In \cite{Benini:2013tra} it was successfully applied also to principal connections for Abelian Yang-Mills models.
We follow here a similar argument in order to motivate the definition of $\tau$. 

Once a gauge invariant functional $\epsilon\in\obsi$ is fixed, 
we are interested in studying how the presence of $\epsilon$ affects the dynamics of the field. 
More precisely, we want to compare the retarded and the advanced effect 
produced by $\epsilon$ on any other gauge invariant functional $\zeta\in\obsi$. 
For each on-shell configuration $h$, this is achieved by finding solutions 
$h^\pm_\epsilon$ to the field equation modified by the presence of $\epsilon$ 
such that $h^\pm_\epsilon$ is gauge equivalent to $h$ in the past/future of a Cauchy surface. 
In the end, the effect produced by $\epsilon$ on $\zeta$ is evaluated 
comparing $\obs_\zeta(h^+_\epsilon)$ with $\obs_\zeta(h^-_\epsilon)$. 
We define the modified dynamics introducing the equation $K_\epsilon h=Kh+2\epsilon^\flat=0$. 
This is exactly the inhomogeneous differential equation 
we would obtain starting from the Lagrangian density for linearized gravity, 
adding an external source $\epsilon$ and then looking for the associated Euler-Lagrange equations. 
In particular this motivates the factor $2$, 
which is due to the fact that $Kh=0$ coincides with eq. \eqref{eqLinGrav} up to such factor. 
Solutions to the equation $K_\epsilon h=0$ can be obtained 
applying the Green operators $E^\pm$ for $P$ to $2\epsilon^\flat$:
\begin{equation*}
K E^\pm(2\epsilon^\flat)=-2\epsilon^\flat+4I\nabla_S\div I E^\pm\epsilon^\flat
=-2\epsilon^\flat+4I\nabla_S E_\Box^\pm(\div\epsilon)^\flat=-2\epsilon^\flat,
\end{equation*} 
where we employed both Lemma \ref{lemUseful} and $\div\epsilon=0$. 
Consider an on-shell configuration $h$, namely $Kh=0$, and look for $h^\pm_\epsilon$ as above. 
Setting $h^\pm_\epsilon=h+ E^\pm(2\epsilon^\flat)$, we read $K_\epsilon h^\pm_\epsilon=0$. 
Moreover $h^\pm_\epsilon$ differs from $h$ only on $J_M^\pm(\supp(\epsilon))$. 
Since $\epsilon$ has compact support, the requirement on the asymptotic behaviour is fulfilled as well. 
We are now ready to define the retarded/advanced $E^\pm_\epsilon$ effect induced by $\epsilon$ 
on any gauge-invariant functional $\zeta\in\obsi$ as 
$E^\pm_\epsilon\zeta=\obs_\zeta(h^\pm_\epsilon)-\obs_\zeta(h)$. 
We stress that, given $h$, the right-hand-side does not depend on our construction of $h^\pm_\epsilon$ 
due to the gauge invariance of $\zeta$. We now compare retarded and advanced effects:
\begin{equation*}
E^+_\epsilon\zeta-E^-_\epsilon\zeta=(E^+(2\epsilon^\flat),\zeta)-(E^-(2\epsilon^\flat),\zeta)
=2(E\epsilon^\flat,\zeta)=\tau([\epsilon],[\zeta]).
\end{equation*}
Therefore Peierls' method yields exactly the presymplectic form used in Proposition \ref{prpPresympl}.
\end{rem}

To conclude the section we establish an isomorphism between the presymplectic space of classical observables $\clobs$ and spacelike compact solutions of the linearized Einstein's equation up to spacelike compact gauge. Besides making contact with other treatments, see {\it e.g.} \cite{Fewster:2012bj}, this correspondence will be exploited in the next section to construct the bulk-to-boundary correspondence. Let us first introduce some notation: We use the symbol $\solsc\doteq\kersc(K)$ to indicate the space of solutions $h$ of the equation $Kh=0$ with support included in a spacelike compact region, while we denote the space of spacelike compact gauge transformations $\nabla_S\chi$, $\chi\in\ssc(T^\ast M)$, with $\gaugesc\doteq\imsc(\nabla_S)$. 

\begin{propo}\label{prpObsSCSol}
There exists a one-to-one correspondence between $\clobs$ and $\solsc/\gaugesc$ induced by the causal propagator $E$ for $P$, which is defined by $[\epsilon]\mapsto[E\epsilon^\flat]$. Such map induces an isomorphism of presymplectic spaces when $\solsc/\gaugesc$ is endowed with the presymplectic form $\sigma:\solsc/\gaugesc\times\solsc/\gaugesc\to\bR$ defined by 
$$\sigma([h], [E\zeta])=2(h,\zeta^\sharp)$$
for each $h\in\solsc$ and $\zeta\in\sc(S^2T^\ast M)$.
\end{propo}
\begin{proof}
The first part of the proof is a slavish copy of that of Proposition \ref{prpEFSol} and of Theorem \ref{thmSolModEauge} where timelike compact sections are replaced by compactly supported ones whereas smooth solutions of the linearized Einstein's equations are replaced by those whose support is spacelike compact and the same is done for gauge transformations. 

For the second part of the proof let us consider $[\epsilon],[\epsilon^\prime]\in\clobs$. 
Then we have 
$$\sigma([E\epsilon^\flat],[E\epsilon^{\prime\,\flat}])=
2(E\epsilon^\flat,\epsilon^\prime)=\tau([\epsilon],[\epsilon^\prime]).$$
This identity completes the proof.
\end{proof}

\begin{rem}\label{remTTSC} 
We may wonder whether the content of Remark \ref{remTTTC}, concerning the implementation of the TT gauge, can be applied also to  $\solsc/\gaugesc$. The point is the following: It is possible to achieve the TT gauge for a spacelike compact de Donder solution $h=E\epsilon$, 
$\epsilon\in\sc(S^2T^\ast M)$, exploiting the spacelike compact gauge freedom $\nabla_S\chi$, 
$\chi\in\ssc(T^\ast M)$ if and only if $\tr\epsilon=\delta\lambda$ for a $\lambda\in\fc^1(M)$, 
that is to say $[\ast(\tr\epsilon)]$ is the trivial class in $\hc^4(M)$. 
Therefore, the obstruction to the TT gauge is now ruled by $\hc^4(M)$, 
which is isomorphic to $\h^0(M)\simeq\bR^c$ via Poincar\'e duality, 
$c$ being the number of connected components of $M$. 
In particular, this means that on every spacetime one may encounter obstructions 
in imposing the TT gauge for some spacelike compact solutions.
\end{rem}


\section{\label{sec3}The bulk-to-boundary correspondence for linearized gravity}
Our present goal is to spell out explicitly the construction of a bulk-to-boundary correspondence for linearized gravity on asymptotically flat spacetimes at a classical level. The quantum counterpart will be discussed in the next section.

\subsection{\label{sec3.1}The phase space on null infinity}
We focus our attention on a particular class of manifolds which are distinguished since they possess an asymptotic behaviour along null directions which mimics that of Minkowski spacetime. Used extensively and successfully in the definition of black hole regions \cite{Wald}, the most general class of asymptotically flat spacetimes includes several important physical examples, such as for instance the Schwarzschild and the Kerr solutions to Einstein's equations. In this paper we will employ the definition of asymptotic flatness, as introduced by Friedrich in \cite{Friedrich:1986uq}. To wit, we consider an \textbf{asymptotically flat spacetime with future time infinity $i^+$}, \textit{i.e.} a globally hyperbolic spacetime $(M, g)$, solution of Einstein's vacuum equations, hereby called \emph{physical spacetime}, such that there exists a second globally hyperbolic spacetime $(\widetilde M, \widetilde g)$, called \emph{unphysical spacetime}, with a preferred point $i^+\in\widetilde M$, a diffeomorphism $\psi : M \to \psi(M) \subset \widetilde M$ and a function $\Xi : \psi(M) \to (0, \infty)$ so that $\psi^\ast(\Xi^{-2}\widetilde g) = g$.
Moreover, the following requirements ought to be satisfied:
\begin{itemize}
  \item[a)]
    If we call $J_{\widetilde M}^-(i^+)$ the causal past of $i^+$, this is a closed set such that $\psi(M) = J_{\widetilde M}^-(i^+) \setminus \partial J_{\widetilde M}^-(i^+)$ and we have $\partial M = \partial J_{\widetilde M}^-(i^+) = \mathscr{I}^+ \cup \{i^+\}$, where $\mathscr{I}^+$ is called \emph{future null infinity}.
  \item[b)]
    $\Xi$ can be extended to a smooth function on the whole $\widetilde M$ and it vanishes on $\mathscr{I}^+ \cup \{i^+\}$.
    Furthermore, $d \Xi \neq 0$ on $\mathscr{I}^+$ while $d \Xi = 0$ on $i^+$ and $\widetilde \nabla_\mu \widetilde \nabla_\nu \, \Xi = -2 \, \widetilde g_{\mu\nu}$ at $i^+$.
  \item[c)]
    Introducing $n^\mu \doteq \widetilde \nabla^\mu \Xi$, there exists a smooth and positive function $\xi$ supported at least in a neighbourhood of $\mathscr{I}^+$ such that $\widetilde \nabla_\mu (\xi^4 n^\mu) = 0$ on $\mathscr{I}^+$ and the integral curves of $\xi^{-1}n$ are complete on future null infinity.
\end{itemize}
Here $\widetilde \nabla$ is the Levi-Civita connection built out of $\widetilde g$. Notice that, in the above definition, future timelike infinity plays a distinguished role, contrary to what happens in the more traditional definition of asymptotically flat spacetimes where $i^+$ is replaced by $i_0$, spatial infinity -- see for example \cite[Section 11]{Wald}. The reason for our choice is motivated by physics: We are interested in the algebra of observables for linearized gravity which is constructed out of $E$, the causal propagator associated to the operator $P$ as in \eqref{eqLinGravDD}. This entails, that, for any smooth and compactly supported symmetric rank $2$ tensor $\epsilon$, its image under the action of the causal propagator is supported in the causal future and past of $\supp(\epsilon)$. Therefore it will be important in our investigation that future timelike infinity is actually part of the unphysical spacetime, so to be able to control the behaviour of $E(\epsilon)$ thereon. Such requirement can be relaxed particularly if one is interested in studying field theories on spacetimes like Schwarzschild where $i^+$ cannot be made part of the unphysical spacetime. The price to pay in this case is the necessity to make sure that any solution of the classical dynamics falls off sufficiently fast as it approaches future timelike infinity. This line of reasoning has been pursued in \cite{Dappiaggi:2009fx}, though we shall not follow it here since it relies heavily on the fact that a very specific manifold has been chosen. On the contrary we plan to consider all at the same time a large class of backgrounds.

Before focusing our attention on the field theoretic side, it is worth devoting a few lines to outlining the geometric properties of the null boundary of an asymptotically flat spacetime. Notice that the choice to work with $\Im^+$ and not with $\Im^-$, past null infinity, is purely conventional. Everything can be translated slavishly to the other case. Here we will summarize what has been already discussed in detail in \cite{Friedrich:1986uq, Geroch, Wald} and in \cite{Dappiaggi:2005ci, Dappiaggi:2011cj, Siemssen:2011gma} for an application to quantum field theory:
\begin{itemize}
\item $\Im^+$ is a three dimensional submanifold of $\widetilde M$ generated by the null geodesics emanating from $i^+$, {\it i.e.} the integral curves of $n$. It is thus diffeomorphic to $\mathbb{R} \times \mathbb{S}^2$ although the possible metric structures are affected by the existence of a gauge freedom which corresponds to the rescaling of $\Xi$ to $\xi \Xi$, where $\xi$ is a smooth function which is strictly positive in $\psi(M)$ as well as in a neighbourhood of $\Im^+$.
\item Null infinity is said to be both intrinsic and universal. In other words, if we introduce for any fixed asymptotically flat spacetime $(M, g)$ the set $C$ composed by the equivalence classes of triples $(\Im^+, h, n)$, where $h \doteq \widetilde{g} \restriction_{\Im^+}$ and $(\Im^+, h, n) \sim (\Im^+, \xi^2 h, \xi^{-1} n)$ for any choice of $\xi$ satisfying c), there is no physical mean to select a preferred element in $C$.
This is called the \emph{intrinsicness} of $\Im^+$. Concerning \emph{universality}, if we select any pair of asymptotically flat spacetimes, $(M_1, g_1)$ and $(M_2, g_2)$, together with the corresponding triples, say $(\Im^+_1, h_1, n_1)$ and $(\Im^+_2, h_2, n_2)$, there always exists a diffeomorphism $\gamma : \Im^+_1 \to \Im^+_2$ such that $h_1 = \gamma^* h_2$ and $n_2 = \gamma_* n_1$.
\item In each equivalence class, element of $C$, there exists a choice of conformal gauge $\xi_B$ yielding a coordinate system $(u, \Xi, \theta, \varphi)$ in a neighbourhood of $\Im^+$, called \emph{Bondi frame}, such that the (rescaled) unphysical metric tensor becomes
\begin{equation}\label{Bondimetric}
  \widetilde g \restriction_{\Im^+} = -2\, \d u\, \d\Xi + \d\theta^2 + \sin^2\!\theta\,\d\varphi^2.
\end{equation}
In this novel coordinate system future null infinity is the locus $\Xi = 0$, while $u$ is the affine parameter of the null geodesics generating $\Im^+$. Thus, at each point on $\Im^+$ the vector field $n$ coincides with $\partial_u$.
\item A distinguished role both from a geometric and from a quantum theoretical point of view is played by the subgroup of diffeomorphisms of $\Im^+$ which maps each equivalence class lying in $C$ into itself. This is the so-called {\bf Bondi-Metzner-Sachs} (BMS) group which coincides, moreover, with the group of asymptotic symmetries of the physical spacetime $(M,g)$ \cite{Geroch}. It can be explicitly characterized in a Bondi frame as follows: Consider the complex coordinates $(z,\bar{z})$ obtained from $(\theta,\varphi)$ via the stereographic projection, $z=e^{i\varphi}\cot(\theta/2)$. An element of the BMS group acts on $(u,z,\bar{z})$ as the following map 
\begin{equation}\label{BMS}
\left\{\begin{array}{l}
u\mapsto u^\prime\doteq K_\Lambda(z,\bar{z})\left(u+\alpha(z,\bar{z})\right),\\
z\mapsto z^\prime\doteq\displaystyle{\frac{az+b}{cz+d}}\;\textrm{and c.c.},
\end{array}\right.
\end{equation}
where $a,b,c,d\in\mathbb{C}$ with $ad-bc=1$, whereas $\alpha(z,\bar{z})\in\C(\mathbb{S}^2)$ and
$$K_\Lambda(z,\bar{z})=\frac{1+|z|^2}{|az+b|^2+|cz+d|^2}.$$
From \eqref{BMS} it descends that the BMS group is the semidirect product $SO_0(3,1)\rtimes \C(\mathbb{S}^2)$, where $SO_0(3,1)$ is the component connected to the identity of the Lorentz group, here acting on $(z,\bar{z})$ via a M\"obius transformation. Hence each $\gamma\in BMS$ identifies actually a pair $\left(\Lambda,\alpha\right)\in SO_0(3,1)\times \C(\bS^2)$.
\end{itemize}

We can now focus on the main goal of this section, namely constructing a $*$-algebra intrinsically defined on $\Im^+$ on which to encode the information of the bulk counterpart in a sense specified below. We start by defining a suitable ``space of observables'' on null infinity and, to this end, we follow a strategy very similar to the one employed in \cite{Dappiaggi:2011cj, Siemssen:2011gma}, which we combine with earlier analysis, see \cite{Ashtekar:1981hw, Ashtekar:1982aa} in particular. Notice that, on account of the peculiar structure of $\Im^+$, it is more convenient to write explicitly the indices of all tensors involved in our analysis. Furthermore we need the following key ingredients:
\begin{enumerate}
\item Let $(M,g)$ be an asymptotically flat spacetime whose unphysical counterpart is $(\widetilde M,\widetilde g)$. Then we call  $\iota:\Im^+\to\widetilde M$ the embedding of null infinity into the unphysical spacetime. 
\item Let $q=\iota^*\widetilde g$ be the pull-back of \eqref{Bondimetric} to $\Im^+$. On account of the null direction on $\Im^+$, the outcome is degenerate and, thus, a canonical inverse metric does not exist. Hence, following the historical convention -- see for example \cite{Ashtekar:1981hw, Ashtekar:1982aa}, we shall call $q^{ab}$ any symmetric tensor field satisfying the condition $q^{ab}q_{ac}q_{bd}=q_{cd}$. There is a large freedom in this choice, but, as we shall comment later, it does not play a role in our analysis.
\end{enumerate}
With these data and on account of the analysis of Ashtekar on the radiative degrees of freedom in general relativity at null infinity, we introduce the following space of sections: 
\begin{equation}\label{symplbound}
\mathcal{S}(\Im^+)\doteq\{\lambda_{ab}\in\s(S^2T^*\Im^+)\;|\; \lambda_{ab}n^a=0,\; \lambda_{ab}q^{ab}=0,\;\textrm{and}\;(\lambda,\lambda)_\Im<\infty,\;(\partial_u \lambda,\partial_u \lambda)_\Im<\infty\}.
\end{equation}
Here $(\cdot,\cdot)_\Im$ denotes a pairing between sections of $S^2T^\ast\Im^+$ defined by
\begin{equation}\label{innprodscri}
(\lambda,\lambda^\prime)_\Im\doteq\int\limits_{\Im^+} \langle\lambda,\lambda^\prime\rangle\,\mu_\Im,\quad \langle\lambda,\lambda^\prime\rangle\doteq \lambda_{ab}\lambda^{\prime}_{cd}q^{ac}q^{bd},
\end{equation}
for all $\lambda,\lambda^\prime\in \s(S^2T^*\Im^+)$ such that $\lambda_{ab}n^a=\lambda^\prime_{ab}n^a=0$ and $\langle\lambda,\lambda^\prime\rangle$ is an integrable function, where, in a Bondi reference frame, $\mu_\Im=\sin^2\theta\,\d u\,\d\theta\,\d\varphi$. Notice that the inner product on $\Im^+$ (and therefore the integrability condition too) does not depend on the choice of $q^{ab}$ since the freedom in the choice of an inverse to $q_{ab}$ lies in the null direction, but the constraint $\lambda_{ab}n^a=0$ on elements of $\mathcal{S}(\Im^+)$ ensures that such components never contribute. Furthermore we can regard $\mathcal{S}(\Im^+)$ as a symplectic vector space after endowing it with the antisymmetric bilinear form below:
\begin{equation}\label{symplbound2}
\sigma_\Im(\lambda,\lambda^\prime)=\int\limits_{\Im^+} 
\left(\lambda_{ab}\mathcal{L}_n\lambda^\prime_{cd}-\lambda^\prime_{ab}\mathcal{L}_n\lambda_{cd}\right)
q^{ac}q^{bd}\,\mu_\Im,
\end{equation}
where $\mathcal{L}_n$ is the Lie derivative along the null vector $n$. Notice that, repeating verbatim the same argument of \cite[Proposition 4.1]{Dappiaggi:2011cj}, \eqref{symplbound2} is weakly non-degenerate as a consequence of the finiteness of both $(\lambda,\lambda)_\Im$ and $(\partial_u\lambda,\partial_u\lambda)_\Im$.

We remark that \eqref{symplbound} differs slightly from the one used in \cite{Ashtekar:1982aa} since an explicit condition on the square integrability of the symmetric $(0,2)$-tensors on $\Im^+$ and of their derivatives along the null-direction is spelled out. This is motivated by our desire to mimic the same analysis for the conformally coupled scalar field in \cite{Moretti:2005ty, Moretti:2006ks} and for the vector potential in \cite{Dappiaggi:2011cj, Siemssen:2011gma}. Furthermore, also the authors in \cite{Ashtekar:1982aa} stress that a suitable fall-off condition of both $\lambda$ and $\lambda^\prime$ towards $i^+$, timelike infinity, is necessary so to ensure the finiteness of the integral in \eqref{symplbound2}.

To conclude our excursus on the boundary data, we need to specify how the BMS group acts on our fields on $\Im^+$.
Following \cite[Section 2]{Dappiaggi:2011cj}, we consider any vector bundle $E$ on $\Im^+$ and we introduce a family of representations of the BMS group on the smooth sections $\s(E)$ via the map $\Pi^\rho: \textrm{BMS} \times \s(E) \to \s(E)$ defined according to 
\begin{equation}\label{BMSrep}
(\Pi^\rho(\gamma,s))(u^\prime,z^\prime,\bar{z}^\prime)
=K_\Lambda(z,\bar{z})^\rho\,s(u+\alpha(z,\bar{z}),z,\bar{z}),
\end{equation}
where $\rho\in\bR$ is a parameter, $\gamma$ is presented in the form $(\Lambda,\alpha)\in SO_0(3,1)\times \C(\bS^2)$ (see the comment below \eqref{BMS}) and $(u^\prime,z^\prime,\bar{z}^\prime)$ are defined as functions of $(u,z,\bar{z})$ according to \eqref{BMS}. With these data we can prove the following:
\begin{propo}\label{BMSinv}
The symplectic space $\left(\mathcal{S}(\Im^+),\sigma_\Im\right)$ is invariant under the representation $\Pi^\rho$ of the BMS group with $\rho=1$.
\end{propo}

\begin{proof}
To start with, let us consider $\Im^+$ in the Bondi frame and let us pick any element $\gamma=(\Lambda,\alpha)\in\textrm{BMS}$. Let $\lambda\in\mathcal{S}(\Im^+)$ and let $\lambda_\gamma\doteq\Pi^\rho(\gamma,\lambda)$ be the outcome after the action of $\gamma$ on $\lambda$ for a given value of $\rho\in\mathbb{R}$. Per construction it holds that $\lambda_\gamma$ is symmetric and it fulfils both $\lambda_{\gamma\;ab}n^a=0$ and $\lambda_{\gamma\;ab}q^{ab}=0$, as these properties are inherited directly from $\lambda$. Furthermore the measure on null infinity can be rewritten in terms of the complex coordinates as $\mu_\Im = -2i(1+\bar{z}z)^{-2}\,\d u\,\d z\,\d\bar{z}$ and, thus, \eqref{BMS} entails that $\mu_\Im$ is transformed by $\gamma$ into $K_\Lambda(z,\bar{z})^3\,\mu_\Im$. Since $K_\Lambda(z,\bar{z})$ is a smooth, bounded and strictly positive function and since $\mu_\Im$ is translation invariant along the $u$-direction, finiteness of both $(\lambda_\gamma,\lambda_\gamma)_\Im$ and $(\partial_u \lambda_\gamma,\partial_u \lambda_\gamma)_\Im$ can be traded from the same property of $\lambda$. In other words, for each $\gamma\in\textrm{BMS}$, $\Pi^\rho(\gamma,\cdot)$ maps $\mathcal{S}(\Im^+)$ to itself, for all $\rho\in\bR$. It remains to be proven that the symplectic form is preserved under action of $\Pi^1(\gamma,\cdot)$ for each $\gamma\in\textrm{BMS}$. We notice that \eqref{symplbound2} can be written as 
$$\sigma_\Im(\lambda,\lambda^\prime)=\int\limits_{\Im^+} \left(\lambda_{ab}\partial_u \lambda^\prime_{cd}-\lambda^\prime_{ab}\partial_u \lambda_{cd}\right)q^{ac}q^{bd}\,\mu_\Im.$$
Let us take any $\gamma=(\Lambda,\alpha)\in\textrm{BMS}$.  At the same time, in a Bondi frame the line element reads $ds^2=0\cdot du^2+d\theta^2+\sin^2\theta d\varphi^2$, from which one can infer via \eqref{BMS} that any BMS group element transforms $q^{ab}$ in $K_\Lambda(z,\bar{z})^{-2}q^{ab}$. Furthermore, for any $\lambda\in\mathcal{S}(\Im^+)$ and for any $\gamma=(\Lambda,\alpha)\in\textrm{BMS}$, it holds that $\partial_u \lambda(u,z,\bar{z})$ is transformed to $\partial_u \lambda(u+\alpha(z,\bar{z}),z,\bar{z})$. Gathering all data together, we infer that, under the action of the BMS group
$$\sigma_\Im(\lambda,\lambda^\prime)\mapsto\int\limits_{\Im^+} K^3_\Lambda\left(K_\Lambda \lambda_{ab}\partial_u \lambda^\prime_{cd}-K_\Lambda \lambda^\prime_{ab}\partial_u\lambda_{cd}\right)K^{-4}_\Lambda q^{ac}q^{bd}\,\mu_\Im=\sigma_\Im(\lambda,\lambda^\prime),$$
where for notational simplicity we have omitted the explicit dependence on the coordinates of the various factors in the integrand. 
\end{proof}

\subsection{\label{sec3.2a}Radiative degrees of freedom and the Geroch-Xanthopoulos gauge}
The next goal of our analysis is to show whether there exists a linear map from the space of classical observables in the bulk to $\mS(\Im^+)$, which is injective and preserves the (pre-)symplectic form. Such result can be extended directly to the quantum counterpart and used to construct Hadamard states induced from null-infinity. It is worth stressing that such procedure has been shown to work for massless and conformally coupled scalar fields \cite{Dappiaggi:2005ci}, for the Dirac field \cite{Dappiaggi:2010gt, Hack:2010iw} and for the vector potential \cite{Dappiaggi:2011cj, Siemssen:2011gma}. We will prove that linearized gravity behaves in an inherently different way from all other free fields. 

The first deviation from the other cases considered manifests itself when, in a given spacetime $(M,g)$, one starts to study the behaviour of \eqref{eqLinGrav} under conformal transformations of the metric. While, in all other scenarios, conformal invariance was guaranteed, in our case a lengthy and tedious calculation not only shows the lack of it but also the arise of terms proportional to inverse powers of $\Xi$, the conformal factor. These are potentially pernicious since we are interested in considering solutions of the equations of motion propagating to null infinity, which is defined as the locus for which $\Xi=0$. To highlight the problem, it is better to work explicitly with indices. Let us consider any solution $h$ of \eqref{eqLinGrav} and let us set $\tau_{ab}=\Xi h^\prime_{ab}$, $\tau_a=\Xi^{-1} n^b h^\prime_{ab}$, $\tau=\widetilde g^{ab}\tau_{ab}$, $f=\Xi^{-1}n^a n_a$ and $n_a=\nabla_a\Xi=\widetilde\nabla_a\Xi$. On $(\widetilde M,\widetilde g)$ it holds that $\tau_{ab}$ obeys to the following partial differential equation -- see \cite[Section 3]{Hollands:2003ie} for an expression valid in all dimensions:
\begin{align}
-\widetilde\Box\tau_{ab}+\widetilde\nabla_a\widetilde\nabla_b\tau+4\nabla_{(a}\tau_{b)}+2\widetilde\nabla_{(a}y_{b)}-2\widetilde g_{ab}\widetilde\nabla^c\tau_c-2\widetilde R_{acdb}\tau^{cd}+\frac{1}{2}\widetilde R_{ab}\tau & \notag\\
+\frac{1}{6}\widetilde R\tau_{ab}-\frac{1}{12}\widetilde g_{ab}\widetilde R\tau+\frac{2}{\Xi}n_{(a}y_{b)}+\frac{\widetilde g_{ab}}{\Xi}\left(2n_c\tau^c+n^c\widetilde\nabla_c\tau+\frac{1}{2}f\tau\right) & =0,\label{monster}
\end{align}
where we used the auxiliary quantity $y_a\doteq\widetilde\nabla^b\tau_{ab}-\widetilde\nabla_a\tau-3\tau_a$.

In order to avoid the ensuing singularities in the above expression, one can only follow the same approach used when dealing with the vector potential in asymptotically flat spacetimes: exploit the gauge invariance of \eqref{eqLinGrav} in order to tame the unwanted terms. Despite the approach is morally the same, notice the sharp contrast with \cite{Dappiaggi:2011cj, Siemssen:2011gma}, in which, on account of the conformal invariance of the Maxwell equations for the Faraday tensor, it suffices to work in the standard Lorenz gauge, albeit it is not conformally invariant. One hopes that similarly one could consider either the de Donder gauge or the transverse-traceless gauge for linearized gravity. Earlier investigations, see for example \cite{Grishchuk:1980hb}, show that pathologies with these choices cannot be avoided. A way to circumvent these problems has been devised by Geroch and Xanthopoulos in \cite{Geroch:1978ur} by introducing a highly non trivial gauge fixing which cancels all potentially divergent terms. The goal of their investigation was to prove the stability under metric perturbations of the property of asymptotic simplicity, which is slightly more general than that of asymptotic flatness. We will review critically this procedure and we will show that, despite the common belief, it cannot be always applied and obstructions are present. While also in electromagnetism similar features are present \cite{Benini:2013ita, Benini:2013tra, DL, SDH12}, here the situation is different since the source of such obstructions cannot be only ascribed to the topology of the spacetime but there is a non trivial interplay with the geometry of the background. The latter plays a key role since all observables of linearized gravity have vanishing divergence, a condition which explicitly involves both the covariant derivative and the metric. 

Let $(M,g)$ be an asymptotically flat spacetime in the sense specified in Section \ref{sec3} and let $(\widetilde M, \widetilde g)$ be an associated unphysical spacetime. By omitting for the sake of notational simplicity the diffeomorphism $\psi: M\to\psi(M)\subset\widetilde M$, we know that, on $M$, $\widetilde g=\Xi^2 g$, where $\Xi\in \C(\widetilde M)$ is strictly positive in $M$ and vanishing on null infinity.  

\begin{defi}\label{GXgauge}
Let $h\in\ssc(S^2T^*M)$ be any spacelike compact solution of \eqref{eqLinGrav}. We say that $h$ is a solution of the linearized Einstein's equations in the Geroch-Xanthopoulos gauge {\bf (GX-gauge)} if, setting $\tau_{ab}=\Xi h^\prime_{ab}$, $\tau_a=\Xi^{-1} n^b h^\prime_{ab}$, $\tau=\widetilde g^{ab}\tau_{ab}$ and $f=\Xi^{-1}n^a n_a$ and $n_a=\nabla_a\Xi=\widetilde\nabla_a\Xi$, it holds that
\begin{subequations}
\begin{align}
y_a=\widetilde\nabla^b\tau_{ab}-\widetilde\nabla_a\tau-3\tau_a & =0,\label{GXa}\\
\left(n^a \widetilde\nabla_a+\frac{1}{6}\Xi \widetilde R+\frac{3}{2}f\right)\widetilde\Box\tau & = \frac{1}{12}\widetilde R f\tau - \frac{1}{2}\tau\widetilde\Box f-\frac{1}{3}\widetilde R n^a\tau_a+\frac{4}{\Xi} \widetilde C_{abcd}\tau^{bd}n^an^c , \label{GXb}
\end{align}
\end{subequations}
where $\widetilde\cdot$ refers to quantities computed with respect to $\widetilde g$, {\em e.g.} $\widetilde C_{abcd}$ is the Weyl tensor for $\widetilde g$. 
\end{defi}

\noindent As proven in \cite{Geroch:1978ur}, working with the GX-gauge reduces \eqref{monster} to the following set of equations: 
\begin{align}\label{monsterGX}
\widetilde\Box\tau_{ab}= & \widetilde\nabla_a\widetilde\nabla_b\tau+ 4\widetilde\nabla_{(a}\tau_{b)}-2\widetilde C_{acbd}\tau^{cd}-\frac{\widetilde R}{6}\tau_{ab}+\frac{\widetilde R}{12}\tau\widetilde g_{ab}-\frac{\tau}{2}\widetilde R_{ab}+2\widetilde{R}_{c(a}\tau^c_{b)}-2 u\widetilde g_{ab},\\
\frac{\widetilde\Box\tau_a}{2}= & \widetilde\nabla_a\sigma\!+\!\frac{\widetilde R_{ac}}{4}\widetilde\nabla^c\tau\!+\!\frac{R}{24}\widetilde\nabla_a\tau\!-\!\frac{\widetilde{R}^{cd}}{2}\widetilde\nabla_c\tau_{ad}-\!\!\frac{\tau_{ac}}{6}\widetilde\nabla^c\widetilde R+\tau^{cd}\widetilde\nabla_{[c}\widetilde R_{a]d}+\tau^c\widetilde R_{ac}+\frac{\widetilde R}{4}\tau_a+\frac{\tau}{12}\widetilde\nabla_a\widetilde R,\notag\\
\widetilde\Box\sigma= & -\frac{1}{2}\widetilde R^{cd}\widetilde\nabla_c\widetilde\nabla_d\tau-2\widetilde R^{cd}\widetilde\nabla_c\tau_d-\frac{1}{12}\widetilde\nabla^c\widetilde R\widetilde\nabla_c\tau+\widetilde R\sigma+\frac{\widetilde R^2}{72}\tau-\frac{1}{2}\tau_{cd}\widetilde R^c_m\widetilde R^{dm}-\frac{1}{3}\tau^c\widetilde\nabla_c\widetilde R,\notag
\end{align}
where $\sigma\doteq\Xi^{-1}(n^a\tau_a+\frac{1}{2}n^a\widetilde\nabla_a\tau+\frac{1}{4}f\tau)$.

\begin{rem}\label{hopedieslastbutonlyIresurrect}
espite being quite complicated, this set of equations is rather advantageous not only because no singularity occurs in the coefficients as $\Xi\to 0$, but also because the system, together with \eqref{GXb}, admits a well-posed initial value problem. This feature was first remarked and exploited in \cite{Geroch:1978ur}, but it can be also inferred from the analysis of \cite[Section 5]{Bar:2013}. More precisely, in the language of this paper and with the notation introduced already by Geroch and Xanthopoulos, we can rewrite the whole system as a PDE of the form $QF=0$ where $F$ is a vector whose entries are the following fields $F_1=\tau_{ab}$, $F_2=\tau_a$, $F_3=\sigma$, $F_4=\tau$, $F_5=\nabla_S\tau$ and $F_6=\Box\tau$. The explicit form of $Q$, which can be constructed out of \eqref{monsterGX} and \eqref{GXb}, is not of particular interest. Noteworthy is instead the following: We are dealing with a linear system of mutually coupled equations, each of which has a principal symbol of hyperbolic type, except for one equation of first order whose principal symbol is of the type described in \cite[Definition 5.1]{Bar:2013}. This entails two important properties. On the one hand, for given smooth and compactly supported initial data, the associated solution is unique, smooth and spacelike compact. On the other hand, one can associate to such a system retarded $\widetilde G^+$ and advanced $\widetilde G^-$ Green operators with the same properties enjoyed by $\widetilde E$ as in Section \ref{subGaugeDyn}.
\end{rem}

Notice that for our aims we want to introduce a procedure to map observables from the bulk spacetime $M$ to null infinity $\Im^+$. Recalling Proposition \ref{prpObsSCSol}, any observable $[\epsilon]\in\clobs$ can be equivalently regarded as a spacelike compact solution up to gauge, $[h]\in\solsc/\gaugesc$. Our strategy consists of extending beyond the boundary a representative $h\in[h]$ which fulfils the GX-gauge, namely such that \eqref{GXa} holds. More precisely, we want to exploit the well-posedness of the Cauchy problem for \eqref{monsterGX} together with \eqref{GXb} to extend in a unique way the above mentioned $h$ to $\widetilde M$.
Nonetheless this last statement has a direct consequence on the GX-gauge, namely, we need to make sure that, given any smooth and spacelike compact solution of 
\eqref{eqLinGrav}, there exists a gauge transformation in $\gaugesc$ such that the gauge-transformed solution fulfils also \eqref{GXa}. A key remark in this direction, mentioned partly in \cite{Geroch:1978ur} and more precisely in \cite{Ashtekar:1982aa} is that, indeed for any solution $h$ of \eqref{eqLinGrav} with smooth and compactly supported initial data, one can always find $\chi\in\ssc(T^\ast M)$ such that $h^\prime=h+\nabla_S\chi$ lies in the GX-gauge. We show why there is a loophole in this statement. More precisely the problem lies in \eqref{GXa}. Let us then rewrite it in terms of geometric quantities and operations defined with respect to the physical metric $g$. It holds:
\begin{align*}
y_a(h^\prime)=\Xi^{-1}v_a(h^\prime), && v_a(h^\prime)=\nabla^b h^\prime_{ab}-\nabla_a h^\prime,
\end{align*}
where we have made explicit the dependence from $h^\prime$. In other words, if we consider $h\in\ssc(S^2T^*M)$ which solves \eqref{eqLinGrav} but fails to fulfil \eqref{GXa}, we need to look for $\chi\in\ssc(T^\ast M)$ such that $y_a(h^\prime)=0$, where $h^\prime=h+\nabla_S\chi$. 
This is tantamount to solving 
\begin{equation}\label{GXc}
\nabla^b\nabla_{[b}\chi_{a]}=-v_a(h).
\end{equation}
This equation is supplemented by the identity $\nabla^a v_a(h)=0$, which follows from \eqref{eqLinGrav} by taking its trace. It is not guaranteed that \eqref{GXc} has a solution. In order to realize it, let us rewrite the equation using the differential $\d$ and the codifferential $\delta$, namely
$$\delta\d\chi=2v(h),\quad \delta v(h)=0.$$
The first equality yields on the one hand a partial differential equation for $\chi$. On the other hand it also imposes a constraint on $v$, namely there cannot exist a solution if $v$ is not coexact. Actually, a solution $\chi$ exists if and only if $v$ is coexact. Unfortunately the second equation does only guarantee that $v$ is coclosed. Furthermore it also imperative that $\chi\in\ssc(T^*M)$ is not only a solution of \eqref{GXc}, but also that its restriction to any Cauchy surface is compactly supported.\footnote{Dropping such requirement would milden the obstruction we have pointed out. As a matter of fact, every coclosed $1$-form is coexact without further constraints on the support if $\h^3(M)$ is trivial, which amounts to considering only globally hyperbolic spacetimes with non-compact Cauchy surfaces.} We stress once more that such requirement is paramount to guarantee the applicability of the procedure to uniquely extend $\Xi h^\prime$ to $\tau\in\ssc(S^2T^*\widetilde M)$ such that $\tau|_M=\Xi h^\prime$ by solving the system \eqref{monsterGX} supplemented with \eqref{GXb}. In the next theorem we show which constraint has to be imposed on the solution of \eqref{eqLinGrav} for this to happen:

\begin{theo}\label{GXchar}
Let $\epsilon\in\obsi$. Denote with $E$ the causal propagator for $P=(\Box-2\Riem)I$. 
Then $v(E\epsilon^\flat)\in\delta\fsc^2(M)$ if and only if $\tr\epsilon\in\delta\fc^1(M)$.
\end{theo}

\begin{proof}
Let us start from the sufficient condition. Suppose $\epsilon\in\obsi$ is such that there exists $\beta\in\fc^1(M)$ 
for which $\delta\beta=\tr\epsilon$. From Lemma \ref{lemma_useful} it follows that 
$\tr E\epsilon=-E_\Box\tr\epsilon=-\delta E_{(1)}\beta$. 
Furthermore, since $\div\epsilon=0$, $\div I\epsilon=-(1/2)\nabla_S\tr\epsilon=-(1/2)\d\delta\beta$ 
and hence $\div E\epsilon=\div IEI\epsilon=E_\Box\div I\epsilon=-(1/2)E_{(1)}\d\delta\beta$. 
Therefore we deduce the following:
\begin{align*}
v(E\epsilon^\flat)=\div E\epsilon-\nabla_S\tr E\epsilon=-\frac{1}{2}E_{(1)}\d\delta\beta+\d\delta E_{(1)}\beta
=-\frac{1}{2}\delta\d E_{(1)}\beta,
\end{align*}
which shows that $v(E\epsilon^\flat)$ lies in $\delta\fsc^2(M)$.

It remains to check that the condition is also necessary. To this avail, assume $\epsilon\in\obsi$ is such that 
there exists $\alpha\in\delta\fsc^2(M)$ for which $\delta\alpha=v(E\epsilon^\flat)$. 
Along the same lines of the first part we get 
\begin{align*}
\div E\epsilon & =-\frac{1}{2}E_\Box\nabla_S\tr\epsilon=-\frac{1}{2}\d E_{(0)}(\tr\epsilon),\\
\nabla_S\tr E\epsilon & =-\nabla_SE_\Box\tr\epsilon=-\d E_{(0)}(\tr\epsilon).
\end{align*}
Therefore our hypothesis entails $\delta\alpha=v(E\epsilon^\flat)=(1/2)\d E_{(0)}(\tr\epsilon)$. 
From this identity it follows that $\d\delta\alpha=0$. According to \cite[Section 7.2]{Benini}, 
there exist $\xi\in\fcd^2(M)$ and $\eta\in\fsc^3(M)$ such that $E\xi+\delta\eta=\alpha$, 
hence we deduce that $\delta\xi=(1/2)\d(\tr\epsilon)+\Box\zeta$ for a suitable $\zeta\in\fc^2(M)$. 
Applying $\delta$ to both sides of the last identity we get $\Box(\tr\epsilon+2\delta\zeta)=0$, 
therefore $\tr\epsilon=-2\delta\zeta\in\delta\fc^2(M)$, thus concluding the proof. 
\end{proof}
 
\noindent Theorem \ref{GXchar}, together with the previous discussion, 
prompts us to introduce the following definition:

\begin{defi}\label{rad_obs}
We say that $[\epsilon]\in\clobs$ is a {\bf radiative classical observable for linearized gravity} 
if there exists a representative $\epsilon\in[\epsilon]$ such that $\tr\epsilon=\delta\beta$ 
for a suitable $\beta\in\fc^1(M)$. 
The collection $\radobs$ of these equivalence classes forms a vector subspace of $\clobs$.
\end{defi}

\noindent Notice that the definition does not depend on the choice of the representative since, for any $\epsilon$ of the form $K\alpha$, $\alpha\in\sc(S^2T^*M)$, where $K$ is as in \eqref{eqLinGravOp}, it holds that $v(EK\alpha)\in\delta\fsc^2(M)$. As usual, $E$ denotes the causal propagator for $P=(\Box-2\Riem)I$. 

At this stage we need to answer an important question: Is there a spacetime where $\radobs$ is smaller than $\clobs$? We show that, contrary to what implicitly assumed in \cite{Geroch:1978ur} and \cite{Ashtekar:1982aa} this is indeed possible. Actually we show two explicit cases: Minkowski spacetime where $\radobs=\clobs$ and an axisymmetric spacetime where instead $\radobs\subsetneq\clobs$.

\paragraph{Minkowski spacetime}
Let us thus consider the simplest example of an asymptotically flat spacetime and let us work with the standard global Cartesian coordinates $x^i$, $i=0,...,3$, so that $M\equiv\mathbb{R}^4$ endowed with metric $\eta=\mathrm{diag}(-1,1,1,1)$. Let $\epsilon\in\obsi$ on Minkowski spacetime. Our goal is to prove that $\tr\epsilon=\delta\beta$ with $\beta\in\fc^1(\bR^4)$. An alternative way to rewrite this condition is to require that $\ast(\tr\epsilon)$ is an exact compactly supported $4$-form. On account of the non-degeneracy of the pairing between $\hc^4(\bR^4)$ and $\h^0(\bR^4)$, this is equivalent to state that $\int_{\bR^4}\ast(\tr\epsilon)=0$. To show that this is indeed the case, it is better to work explicitly with indices. Since $\epsilon\in\obsi$, $\partial_a\epsilon^{ab}=0$ for all $b=0,\dots,3$ and thus, for example for the case $b=1$, it holds
\begin{align*}
\partial_1\epsilon^{11}=-\partial_{0}\epsilon^{01}-\partial_2\epsilon^{21}-\partial_3\epsilon^{31}.
\end{align*}
Therefore it follows that 
\begin{align*}
\epsilon^{11}(x^0,x^1,x^2,x^3)
=-\int\limits_{-\infty}^{x^1}\big(\partial_{0}\epsilon^{01}(x^0,y^1,x^2,x^3)
+\partial_2\epsilon^{21}(x^0,y^1,x^2,x^3)+\partial_3\epsilon^{31}(x^0,y^1,x^2,x^3)\big)\,\d y^1.
\end{align*}
Adapting the values of the indices, a similar formula can be written for $\epsilon^{aa}$, $a=0,2,3$. 
We can now compute
$$\int\limits_{\bR^4}\ast(\tr\epsilon)=\int\limits_{\bR^4}\eta_{aa}\epsilon^{aa}\,\d^4x,$$
as the sum of four integrals, which actually do vanish separately. 
To show this, let us just focus on the contribution from $\epsilon^{11}$:
\begin{align*}
\int\limits_{\bR^4}\epsilon^{11}\,\d^4x
& =-\int\limits_{\bR^4}\d^4x\int\limits_{-\infty}^{x^1}\big(\partial_{0}\epsilon^{01}(x^0,y^1,x^2,x^3)
+\partial_2\epsilon^{21}(x^0,y^1,x^2,x^3)+\partial_3\epsilon^{31}(x^0,y^1,x^2,x^3)\big)\,\d y^1.
\end{align*}
Each of the three contributions to $\int_{\bR^4}\epsilon^{11}\d^4x$ vanishes. 
For example, the first contribution can be rewritten as 
\begin{align*}
\int\limits_{\bR^4}\d^4x\int\limits_{-\infty}^{x^1}\partial_{0}\epsilon^{01}(x^0,y^1,x^2,x^3)\,\d y^1
=\int\limits_\bR\d x^1\int\limits_\bR\d x^2\int\limits_\bR\d x^3\int\limits_{-\infty}^{x^1}\d y^1
\int\limits_\bR\partial_{0}\epsilon^{01}(x^0,y^1,x^2,x^3)\,\d x^0=0.
\end{align*}
The integral along the variable of derivation entails evaluation of the components of $\epsilon^{ab}$ at $\pm\infty$. 
On account of the compact support this always vanishes. 
Therefore $\int_{\bR^4}\epsilon^{11}\d^4x$ vanishes as well 
and the same holds true for $\epsilon^{aa}$, $a=0,2,3$, by the same argument. 
Hence $\int_{\bR^4}\ast(\tr\epsilon)=0$ or, in other words, 
all observables for linearized gravity on Minkowski spacetime are of radiative type. 
Notice that in our analysis a key role is played by the geometry of the background. 
Even mild changes in the metric coefficients would invalidate our line of reasoning 
and hence no positive result could be obtained.

\paragraph{Axisymmetric spacetime}
Let $M$ be a globally hyperbolic spacetime, which topologically looks like $\mathbb{R}^3\times\mathbb{S}^1$, and let us consider thereon the standard coordinates $(t,x,y,\varphi)$. Let us suppose that the line element is of the form $\d s^2=g_{ij}\d x^i\d x^j+g_{\varphi\varphi}\d\varphi^2$, $i=t,x,y$, where all coefficients are smooth and independent of $\varphi$. Hence $(M,g)$ admits a Killing field along $\mathbb{S}^1$. Let us further notice that the tangent bundle is trivial and thus it is legitimate to consider the components of any $\epsilon\in\sc(S^2TM)$ as global sections. Since we want to consider an element of $\epsilon\in\obsi$, we set 
\begin{align*}
\epsilon^{\varphi\varphi}=\frac{1}{2\pi\sqrt{|g|}g_{\varphi\varphi}}f(t)f(x)f(y), 
&& \epsilon^{ab}=0\quad\mbox{for}(a,b)\neq(\varphi,\varphi),
\end{align*}
where $|g|$ here stands for the absolute value of the determinant of the metric. Notice that non-degeneracy of the metric entails that $g_{\varphi\varphi}$ is nowhere vanishing. The function $f\in \Cc(\mathbb{R})$ is chosen such that its integral along $\bR$ is equal to $1$. Notice that per construction $\epsilon$ is compactly supported and gauge invariant, namely $\div\epsilon=0$. In fact the only component which might have a non vanishing contribution is $(\div\epsilon)^\varphi$, for which we have $(\div\epsilon)^\varphi=\nabla_\varphi\epsilon^{\varphi\varphi}=0$ since $\epsilon^{\varphi\varphi}$ is independent of $\varphi$ and because of the form of the line element $\d s^2$. 
Let us now consider $\tr\epsilon=g_{\varphi\varphi}\epsilon^{\varphi\varphi}=(2\pi\sqrt{|g|})^{-1}f(t)f(x)f(y)$. In order to show that $\tr\epsilon$ is not coexact, we can use the same argument as in the previous example, namely we compute
$$\int\limits_M\ast(\tr\epsilon)=\int\limits_\bR\int\limits_\bR\int\limits_\bR\int\limits_{\bS^1}\frac{1}{2\pi}f(x)f(y)f(t)\,\d\varphi\,\d y\,\d x\,\d t=1.$$
In other words we have constructed explicitly an equivalence class $[\epsilon]\in\clobs$ which is not of radiative type. We stress that, as far as \eqref{GXc} is concerned, asymptotic flatness or simplicity of the background is not a necessary prerequisite and one could look for solutions of such equation independently. Yet, for the sake of completeness, we mention that axisymmetric asymptotically simple spacetimes are known to exist and they have been extensively studied in the literature, see \cite{Bicak:1984en} and references therein.

\vskip .3cm

To conclude the section, we stress that \eqref{GXb}, that is the residual gauge fixing, is nothing but a rather involved partial differential equation which can be solved with the argument given in the appendix of \cite{Geroch:1978ur} and it does not yield any problem in terms of implementation and support properties. Hence we can slightly adapt the result of \cite{Geroch:1978ur} in order to account for the obstructions written above:
\begin{theo}
Let $(M,g)$ be a an asymptotically flat spacetime whose associated unphysical spacetime is $(\widetilde M, \widetilde g)$ with associated conformal factor $\Xi$. Let $[E\epsilon]\in\solsc/\gaugesc$ be a gauge equivalence class of spacelike compact solutions of the linearized Einstein's equations, where $\epsilon$ is any representative of $[\epsilon]\in\radobs$ and $E$ is the causal propagator for $P=(\Box-2\Riem)I$. Then there exists $h^\prime\in[E\epsilon^\flat]$ which is {\bf asymptotically regular}, that is $\tau=\Xi h^\prime$ admits an extension to $\widetilde M$ whose restriction to $\Im^+$ is smooth. Furthermore both $\tau_{ab}n^a$ and $\tau_{ab}n^a n^b$ do admit a vanishing limit to null infinity.
\end{theo}

We remark that our concept of radiative observables is related to that of radiative degrees of freedom as used for example in \cite{Ashtekar:1981hw} although, in this paper, the focus is on the structure of the kinematical arena on null infinity, whereas we are interested more on those observables which can be mapped to the boundary compatibly with the dynamics and with gauge invariance.

\subsubsection{\label{sec3.2}Projecting to the boundary}

We have all ingredients to define a well-behaved projection of the classical radiative observables from the bulk to the boundary. Our approach extends the one already discussed in \cite{Ashtekar:1982aa} although we take into account the obstruction outlined in Theorem \ref{GXchar}:
\begin{theo}\label{b2b-cl}
Let $\left(\radobs,\tau\right)$ be the space of radiative classical observables as in \eqref{obs} endowed with the presymplectic form defined in Proposition \ref{prpPresympl}. Then there exists a map $\Upsilon:\radobs\to\mathcal{S}(\Im^+)$ defined by 
\begin{equation}\label{proj}
\Upsilon([\epsilon])_{ab}=\gamma_{ab}-\frac{1}{2}\gamma_{cd}q^{cd}q_{ab},
\end{equation}
where $\iota:\Im^+\to\widetilde M$ is the embedding of $\Im^+$ into $\widetilde M$, $\gamma=\iota^*\tau$, $\tau$ is the extension to $\widetilde M$ of $\Xi h^\prime$ obtained solving \eqref{monsterGX} together with \eqref{GXb} and $h^\prime\in[E\epsilon^\flat]$ is a solution of \eqref{eqLinGrav} in the GX-gauge built out of $[\epsilon]$. Furthermore, for all $[\epsilon],[\epsilon^\prime]\in\radobs$, it holds that $\sigma_\Im(\Upsilon[\epsilon],\Upsilon[\epsilon^\prime])=\tau([\epsilon],[\epsilon^\prime])$.
\end{theo}

\begin{proof}
The proof is a recollection of already known results. To start with, on account of the construction of Geroch and Xanthopoulos we know that for all $[\epsilon]\in\radobs$, $\Upsilon([\epsilon])\in\s(S^2T^*\Im^+)$ and $\Upsilon([\epsilon])_{ab}n^a=0$ and $\Upsilon([\epsilon])_{ab}q^{ab}=0$. Furthermore, since every representative of $[\epsilon]$ is compactly supported, the support properties of the causal propagator $E$ as well as the extensibility of the solution to $\widetilde M$  entail that there exists $u_0\in\mathbb{R}$ such that, in the Bondi frame, $\Upsilon([\epsilon])=0$ on $\left(-\infty,u_0\right)\times\bS^2$.  Such property is useful to prove that both $(\tau,\tau)_\Im$ and $(\partial_u\tau,\partial_u\tau)_\Im$ are finite. As a matter of fact we can now following slavishly the same proof used in \cite[Theorem 4.4]{Dappiaggi:2011cj} for the vector potential from which the sought statement descends. Hence $\Upsilon$ maps $\mathcal{E}$ in $\mathcal{S}(\Im^+)$. Furthermore, on account of \cite[Theorem 2]{Ashtekar:1982aa}, $\tau([\epsilon],[\epsilon^\prime])=\sigma_\Im(\Upsilon[\epsilon],\Upsilon[\epsilon^\prime])$.  
\end{proof}

Notice that $\Upsilon([\epsilon])$ has only two independent components compatibly with our expectations on the degrees of freedom for linearized gravity.

\subsection{\label{sec3.3}Hadamard states for linearized gravity}
Goal of this section is to extend the classical bulk-to-boundary correspondence defined in Theorem \ref{b2b-cl} to the quantum level and to exploit the outcome to construct explicitly Hadamard states for linearized gravity. The first part of this programme is rather straightforward with all the building blocks we have. As a starting point we construct the algebra of observables both for the bulk theory and for the one living intrinsically on null infinity.

\begin{defi}
Let $\mathcal{T}(\clobs)$ be the tensor algebra built out of \eqref{obs} as $\mathcal{T}(\clobs)=\bigoplus_n\clobs^{\otimes n}_\mathbb{C},$
where the zeroth-tensor power is nothing but $\bC$ and the subscript $\mathbb{C}$ denotes complexification. We call {\em algebra of observables for linearized gravity} the quotient $\mathcal{F}(\clobs)$ between $\mathcal{T}(\clobs)$ and $\mathcal{I}$, the ideal generated by elements of the form $-i\tau([\epsilon],[\epsilon^\prime])\oplus [\epsilon]\otimes([\epsilon^\prime]-[\epsilon^\prime]\otimes[\epsilon])$, where $\tau$ is the presymplectic form \eqref{tau}. This is a $*$-algebra if endowed with the $*$-operation induced by complex conjugation. At the same time we call {\em algebra of radiative observables for linearized gravity} the algebra $\mathcal{F}(\radobs)$ built replacing $\clobs$ with $\radobs$.
\end{defi}

\noindent Notice that, since $\tau$ is possibly degenerate, $\mathcal{F}(\clobs)$ is not guaranteed to be a simple algebra. In other words it may possess a non trivial center. This feature gives rise to potential problems in interpreting the theory in the framework of the principle of general local covariance, as it has been already thoroughly discussed for Abelian gauge theories \cite{Benini:2013ita, Benini:2012vi, SDH12}. Yet, for the sake of constructing states, central elements do not play a distinguished role. We can define a counterpart of $\mathcal{F}(\radobs)$ on the boundary out of the symplectic space $\mathcal{S}(\Im^+)$ defined in \eqref{symplbound}.

\begin{defi}
Let $\mathcal{T}(\Im^+)$ be the tensor algebra built out of \eqref{symplbound} as $\mathcal{T}(\Im^+)=\bigoplus_n\mathcal{S}(\Im^+)^{\otimes n}_\mathbb{C},$
where the zeroth-tensor power is nothing but $\bC$ and the subscript $\mathbb{C}$ denotes complexification. We call {\em algebra of observables on null infinity} the quotient $\mathcal{F}(\Im^+)$ between $\mathcal{T}(\Im^+)$ and $\mathcal{I}_\Im$, the ideal generated by elements of the form $-i\sigma_\Im(\lambda,\lambda^\prime)\oplus (\lambda\otimes\lambda^\prime-\lambda^\prime\otimes\lambda)$, where $\sigma_\Im$ is the symplectic form \eqref{symplbound2}. This is a $*$-algebra if endowed with the $*$-operation induced by complex conjugation.
\end{defi}

\noindent We can relate the two algebra of observables we built as follows:

\begin{propo}\label{b2b-q}
There exists a $*$-homomorphism $\mathcal{F}(\radobs)\to\mathcal{F}(\Im^+)$ specified on the generators by $[\epsilon]\mapsto\Upsilon[\epsilon]$, where $\Upsilon$ is defined in \eqref{proj}. 
With a slight abuse of notation, we use $\Upsilon$ to denote also the $*$-homomorphism defined here.
\end{propo}

\begin{proof}
In order to prove that $\iota$ is an homomorphism it suffices to show that it preserves the presymplectic form when evaluated on all generators. Both these conditions have been already proven in Theorem \ref{b2b-cl}. To conclude we notice that all operations involved do not affect the complex conjugations and thus we have constructed a $*$-homomorphism.
\end{proof}

\noindent Having set up the bulk-to-boundary correspondence for the algebra of observables, we are ready to discuss the construction of Hadamard states. Let us recall that, for any unital $*$-algebra $\mathcal{A}$, an algebraic state is a map $\omega:\mathcal{A}\to\mathbb{C}$ such that $\omega(a^*a)\geq 0$ for all $a\in\mathcal{A}$ and $\omega(e)=1$, where $e\in\mathcal{A}$ is the identity element. The role of a state is to allow us to recover the standard probabilistic interpretation of a quantum system via the GNS theorem which associates to each pair $(\mathcal{A},\omega)$ a triple $(\mathcal{D}_\omega,\pi_\omega,\Omega_\omega)$ consisting of a dense subspace $\mathcal{D}_\omega$ of a Hilbert space $\mathcal{H}_\omega$, a representation $\pi_\omega$ of $\mathcal{A}$ in terms of linear operators on $\mathcal{D}_\omega$ and a cyclic unit-norm vector $\Omega_\omega$ such that $\mathcal{D}_\omega=\overline{\pi_\omega(\mathcal{A})\Omega_\omega}$. This triple is unique up to unitary equivalence and, for any $a\in\mathcal{A}$, $\omega(a)=\left(\Omega_\omega,\pi_\omega(a)\Omega_\omega\right)_{\mathcal{H}_\omega}$, where $(\cdot,\cdot)_{\mathcal{H}_\omega}$ is the inner product in $\mathcal{H}_\omega$. If the role of $\mathcal{A}$ is played by the algebra of fields, such as in our case either $\mathcal{F}(\clobs)$ or $\mathcal{F}(\radobs)$, we can focus our attention on a special subclass which is often used in theoretical and mathematical physics. We are referring to the quasi-free/Gaussian states which are completely defined in terms of their n-point correlation functions $\omega_n$. In particular $\omega_n=0$ if $n$ is odd whereas, if $n$ is even, than
\begin{equation}\label{Gaussian}
\omega_n(\lambda_1\otimes...\otimes\lambda_n)=\sum\limits_{\pi_n\in S_n}\prod\limits_{i=1}^{n/2}\omega_2(\lambda_{\pi_n(2i-1)}\otimes\lambda_{\pi_n(2i)}),
\end{equation}
where $\lambda_i\in\mathcal{E}$, $i=1,2,...,n$, whereas $S_n$ denotes the ordered permutations of $n$ elements.

The explicit identification of a state for a quantum field theory is usually a rather daunting quest unless the symmetries of the background are sufficient to help us in singling out a preferred candidate, {\em e.g.} the vacuum in Minkowski spacetime, whose existence and uniqueness is a by-product of Poincar\'e invariance. Yet, since a generic curved background might even have a trivial isometry group, one has to look for a different procedure to construct explicitly a quantum state. On the class of globally hyperbolic and asymptotically flat spacetimes, we will show that Proposition \ref{b2b-q} provides a tool to induce states for the bulk theory starting from the boundary counterpart. The advantage of working with a theory defined on null-infinity is two-fold: On the one hand $\Im^+$ is nothing but $\bR\times\mathbb{S}^2$ and thus along the null $\bR$-direction, one can perform a Fourier transform, thus working in terms of modes. On the other hand, as discussed in Section \ref{sec3.1}, a theory on $\Im^+$ is invariant under a suitable action of the BMS group. The latter plays the same role of the Poincar\'e group in Minkowski spacetime in helping us to single out a distinguished state at null infinity. More precisely, along the same lines of \cite{Dappiaggi:2005ci,Dappiaggi:2010gt,Dappiaggi:2011cj}, the following proposition holds true:

\begin{propo}\label{2pt}
The map $\omega_2^\Im:\mathcal{S}(\Im^+)_{\mathbb{C}}\otimes\mathcal{S}(\Im^+)_{\mathbb{C}}\to\mathbb{R}$ such that 
\begin{equation}\label{eq2pt}
\omega_2^\Im(\lambda\otimes \lambda^\prime)=-\frac{1}{\pi}\lim_{\epsilon\to 0}\int\limits_{\mathbb{R}^2\times\mathbb{S}^2}\frac{\lambda_{ab}(u,\theta,\varphi)\lambda^{\prime}_{cd}(u^\prime,\theta,\varphi)q^{ac}q^{bd}}{(u-u^\prime-i\epsilon)^2}\,\d u\,\d u^\prime\,\d\bS^2(\theta,\varphi),
\end{equation}
where $d\bS^2(\theta,\varphi)$ is the standard line element on the unit $2$-sphere, unambiguously defines a quasi-free state $\omega^\Im:\mathcal{F}(\Im^+)\to\mathbb{C}$. Furthermore:
\begin{enumerate}
\item $\omega^\Im$ induces via pull-back a quasi-free bulk state $\omega^M:\mathcal{F}(\radobs)\to\mathbb{C}$ such that $\omega^M\doteq\omega^\Im\circ\Upsilon$,
\item $\omega^\Im$ is invariant under the action\footnote{Here we use the symbol $\Pi$ with a slight abuse of notation since we have already introduced it to indicate in Proposition \ref{BMSinv} the representation of the BMS group on $\mathcal{S}(\Im^+)$. Since $\mathcal{F}(\Im^+)$ is built out of $\mathcal{S}(\Im^+)$ we feel that no confusion can arise.} $\Pi$ of the BMS group induced on $\mathcal{F}(\Im^+)$ by \eqref{BMS}.
\end{enumerate}
\end{propo}

\begin{proof}
As a starting point, we notice that from $\omega_2^\Im$ we can define unambiguously a Gaussian state $\omega^\Im$ on $\mathcal{F}(\Im^+)$ via \eqref{Gaussian}. Yet one needs to ensure positivity of $\omega^\Im$ which is equivalent to showing that this property holds true for $\omega_2^\Im$. To this avail, we notice that every element $\lambda\in\mathcal{S}(\Im^+)$ tends to $0$ as $u$ tends to $\pm\infty$ as one can prove readapting the argument of \cite[Footnote 7, p. 52]{Moretti:2005ty} and exploiting that both $(\lambda,\lambda)_\Im<\infty$ 
and $(\partial_u \lambda,\partial_u\lambda)_\Im<\infty$. Hence we can perform a Fourier-Plancherel transform along the u-direction -- see \cite[Appendix C]{Moretti:2006ks} -- and eventually use the convolution theorem to obtain
$$\omega_2^\Im(\lambda\otimes \lambda^\prime)=\frac{1}{\pi}\lim_{\epsilon\to 0}\int\limits_{\mathbb{R}\times\mathbb{S}^2}k\Theta(k)\widehat{\lambda}_{ab}(k,\theta,\varphi)\widehat{\lambda}^{\prime}_{cd}(-k,\theta,\varphi)q^{ac}q^{bd}\,\d k\,\d\mathbb{S}^2(\theta,\varphi),$$
where $\Theta(k)$ is the Heaviside step function. From this last formula it follows that $\omega_2^\Im(\lambda\otimes\bar{\lambda})\geq 0$ and, moreover, that $\omega_2^\Im(\lambda\otimes\lambda^\prime)-\omega_2^\Im(\lambda^\prime\otimes\lambda)=i\sigma_\Im(\lambda,\lambda^\prime)$ where $\sigma_\Im$ is the symplectic form defined in \eqref{symplbound2}. Hence $\omega^\Im$ is indeed a state on $\mathcal{F}(\Im^+)$. Proposition \ref{b2b-q} ensures that $\omega^M$ is in turn a well-defined and, per construction, quasi-free state on $\mathcal{F}(\radobs)$. Only BMS invariance remains to be proven. Since $\omega^\Im$ is quasi-free, it suffices to show the statement for the two-point function. Let us first collect all ingredients we need. On account of Proposition \ref{BMSinv}, we know already that \eqref{BMS} induces the action $\Pi^1$ on each $\lambda\in\mathcal{S}(\Im^+)$ -- see \eqref{BMSrep}. Furthermore, for each $\gamma=(\Lambda,\alpha)\in\textrm{BMS}$, suppressing the dependence on the coordinates on the $2$-sphere for the BMS group elements and actions, $q^{ab}\mapsto K^{-2}_\Lambda q^{ab}$, whereas $\d u\mapsto K_\Lambda \d u$ and $\d\mathbb{S}^2(\theta,\varphi)\mapsto K^2_\Lambda \d\mathbb{S}^2(\theta,\varphi)$. Putting everything together, it holds that, for every $\lambda,\lambda^\prime\in\mathcal{S}(\Im^+)$
$$(\omega^2_\Im\circ\Pi)(\lambda\otimes\lambda^\prime)=-\frac{1}{\pi}\lim_{\epsilon\to 0}\int\limits_{\mathbb{R}^2\times\mathbb{S}^2}\frac{K_\Lambda\lambda_{ab}(u,\theta,\varphi)K_\Lambda\lambda^{\prime}_{cd}(u^\prime,\theta,\varphi)K^{-2}_\Lambda q^{ac} K^{-2}_\Lambda q^{bd}}{(K_\Lambda u- K_\Lambda u^\prime-i\epsilon)^2}K^4_\Lambda\,\d u\,\d u^\prime\,\d\bS^2(\theta,\varphi),$$
which coincides with $\omega^2_\Im$ since the factor $K_\Lambda$ is bounded and thus we are free to redefine the $\epsilon$-term accordingly.
\end{proof}

In the previous proposition we have displayed a state for the bulk algebra of fields which is constructed out of the boundary counterpart. A genuine question at this stage would be why should one consider such candidate. To answer, notice that, in between the plethora of all possible states for the algebra of fields of a quantum theory, not all of them can be considered as physically sensible. While in Minkowski spacetime, it is always possible to resort for free fields to the Poincar\'e vacuum, which is, moreover, unique, this luxury is not at our disposal on curved backgrounds. The main reason is due to the absence in general of a sufficiently large isometry group. In this case it is paramount to find a criterion to select in between all possible states those which are acceptable. After long debates, there is now a wide consensus that such statement translates in the request that the ultraviolet behaviour mimics that of the Poincar\'e vacuum on Minkowski spacetime and that the quantum fluctuations of all observables, such as, for example, the smeared components of the stress energy tensor, are bounded. From a mathematical point of view, this translates in choosing Hadamard states $\omega$, namely states which satisfy a condition on the singular structure of the bi-distribution $\Omega_2$ associated to their two-point function $\omega_2$. Before spelling it out explicitly, we remark that, as discussed in \cite{Hollands:2001nf} for scalar field theories, Hadamard states  can be used to define a locally covariant notion of Wick products of fields and, thus, interactions can be discussed at least at a perturbative level. Since the main feature which is exploited in the analysis of Hollands and Wald is essentially that in every normal neighborhood of the underlying manifold the singular structure of $\Omega_2$ depends only on the geometry of the spacetime, we expect that their results can be extended also to the case we are considering. Yet we shall not dwell into this topic since it would lead us far from our goal. 

On the contrary we review now the tools necessary to define rigorously Hadamard states, adapting them to the case at hand. All basic notions and definitions related to microlocal analysis are taken as in \cite{hormander:1990bl}, to which we refer. Let us thus consider a bi-distribution $\Omega_2:\sc(S^2TM)\times \sc(S^2TM)\to\bR$ which is furthermore a bi-solution of $\widetilde P=\Box-2\Riem$. Although here $\Omega_2$ is generic and not necessarily stemming from a two-point function of a quasi-free state, we employ the same symbol for the sake of notational simplicity.  The latter could be in principle any normally hyperbolic operator and, not necessarily $\widetilde P=\Box-2\Riem$ as used previously in the text. Since we will be interested only in this case, we feel safe using such notation. Following \cite{Sahlmann:2000zr} and also the discussion after \cite[Theorem 8.2.4]{hormander:1990bl} we know that, for any vector bundle $E$, the wavefront set $WF(u)$ of a distribution $u\in\mathcal{D}^\prime(E)$ is defined locally as the union of $WF(u_i)$, the wavefront set of each component of $u$ in a local trivialization of $E$. Hence we can extend also to this scenario the definition given in \cite{Radzikowski:1996pa, Radzikowski:1996ei, Sahlmann:2000zr}:  
\begin{defi}\label{Had}
A two-point distribution $\Omega_2:\sc(S^2TM)\times \sc(S^2TM)\to\bR$, bi-solution of $\widetilde P=\Box-2\Riem$ is said to be of {\bf Hadamard form} if it satisfies the following conditions:
\begin{itemize}
\item[1.] $WF(\Omega_2)=\{(x,k,x^\prime,k^\prime)\in T^\ast (M\times M)\setminus\{0\}\;|\;(x,k)\sim (x^\prime,-k^\prime),\; k\triangleright 0\},$
where $0$ is the zero section of $T^\ast (M\times M)$, whereas $(x,k)\sim(x^\prime,-k^\prime)$ means that the point $x$ is connected to $x^\prime$ by a lightlike geodesic $\gamma$ so that $k$ is cotangent to $\gamma$ in $x$ and $-k^\prime$ is the parallel transport of $k$ from $x$ to $x^\prime$ via $\gamma$. Furthermore $k\triangleright 0$ means that the covector $k$ is future-directed;
\item[2.] $\Omega_2(\epsilon,\zeta)-\Omega_2(\zeta,\epsilon)=2i(\widetilde{E}\epsilon^\flat,\zeta)$ 
for each $\epsilon,\zeta\in\obsk$, where the equality holds true up to smooth terms which vanish 
when smeared on $\epsilon,\zeta\in\obsi$.
\end{itemize}
\end{defi}

Contrary to what happens for electromagnetism in \cite{Dappiaggi:2011cj, Siemssen:2011gma}, for linearized gravity we need to introduce an additional concept developed in \cite[Chapter 6]{Hunt2012} to cope with the operation of trace-reversal which is present in \eqref{eqLinGravDD}. 

\begin{defi}\label{staterev}
Let $\Omega_2:\sc(S^2TM)\times \sc(S^2TM)\to\bR$ be any bi-distribution. We call {\em trace} of $\Omega_2$ the scalar bidistribution $(\tr\,\Omega_2):\Cc(M)\times \Cc(M)\to\bR$ such that, for all $f,f^\prime\in \Cc(M)$,
$$(\tr\,\Omega_2)(f,f^\prime)\doteq \Omega_2(fg^{-1},f^\prime g^{-1}),$$
where $g^{-1}$ is the inverse of the metric tensor. We call {\em trace reversal} of $\Omega_2$ the bidistribution $I\,\Omega_2:\sc(S^2TM)\times \sc(S^2TM)\to\bR$, such that, for all $\epsilon,\zeta\in\sc(S^2TM)$ 
$$(I\,\Omega_2)(\epsilon,\zeta)\doteq\Omega_2(\epsilon,\zeta)-\frac{1}{8}(\tr\,\Omega_2)(\tr\epsilon,\tr\zeta),$$
where $\tr\epsilon=g_{\mu\nu}\epsilon^{\mu\nu}$.
\end{defi}

Notice that the apparently strange coefficient $1/8$ ensures both that $\tr(I\,\Omega_2)=-\tr\,\Omega_2$ and that $I(I\,\Omega_2)=\Omega_2$. The trace reversal of a bidistribution plays a key role in understanding what is a Hadamard state for linearized gravity. As a matter of fact, the presymplectic form with which $\clobs$ ($\radobs$) is endowed and, accordingly, the canonical commutation relations with which $\mathcal{F}(\clobs)$ ($\mathcal{F}(\radobs)$) is constructed are built out of $E$ the causal propagator of $P=\widetilde{P}I$, where $\widetilde P=\Box-2\Riem$ and $I$ is the trace reversal on $\s(S^2TM)$. As noted in \cite{Hunt2012}, every bisolution $\Omega_2$ of $\widetilde{P}$ is also one for $P$ and the same holds true for $I\,\Omega_2$. Furthermore
acting with $I$ on $\Omega_2$ changes the second requirement of Definition \ref{Had} into a similar one 
where $\widetilde{E}$, the causal propagator for $\widetilde{P}$, is replaced by $E$, 
the causal propagator for $P$. In fact
\begin{align*}
(I\,\Omega_2)(\epsilon,\zeta) & -(I\,\Omega_2)(\zeta,\epsilon)
=\Omega_2(\epsilon,\zeta)-\Omega_2(\zeta,\epsilon)
-\frac{1}{8}(\tr\,\Omega_2)(\epsilon,\zeta)+\frac{1}{8}(\tr\,\Omega_2)(\zeta,\epsilon)\\
& =2i(\widetilde{E}\epsilon^\flat,\zeta)-\frac{1}{4}i(\widetilde{E}(g^{-1}\tr\epsilon)^\flat,g^{-1}\tr\zeta)
=2i(\widetilde{E}\epsilon^\flat,\zeta)-i(g\tr\widetilde{E}\epsilon^\flat,\zeta)
=2i(E\epsilon^\flat,\zeta).
\end{align*}

We can now define the notion of a Hadamard state for linearized gravity, although we have still to take care of the additional constraint that $\clobs$ is generated by compactly supported sections with vanishing divergence. Taking into account also Definition \ref{Had}, we recall the definition of Hadamard states for linearized gravity given in \cite{Hunt2012}.

\begin{defi}
A quasi-free state $\omega:\mathcal{F}(\clobs)\to\mathbb{C}$ is said to be a {\bf Hadamard state} if there exists a bi-distribution $\Omega_2:\sc(S^2TM)\times\sc(S^2TM)\to\mathbb{C}$, which is a bi-solution of $\widetilde P$ of Hadamard form (see Definition \ref{Had}) such that, for every $\epsilon,\zeta\in\obsi$, 
$$\omega([\epsilon]\otimes[\zeta])=(I\,\Omega_2)(\epsilon,\zeta),$$
where $I$ is as in Definition \ref{staterev}. A similar definition holds with $\clobs$ replaced by $\radobs$ simply restricting to those gauge invariant functional $\epsilon\in\obsi$ which fulfil the requirement of Definition \ref{rad_obs}, namely such that $\tr\epsilon=\delta\beta$ for a suitable $\beta\in\fc^1(M)$. 
\end{defi}

As a last step, we have to show that the state for $\mathcal{F}(\clobs)$ constructed via the bulk-to-boundary procedure fits indeed in the class characterized in the previous definition. In view of the previous definitions we cannot work directly with $\omega^M$ as in Proposition \ref{2pt} and prove that it is of Hadamard form. To this end we would need to show that there exists a bi-distribution on $M$, hence defined on all pairs of elements in $\sc(S^2TM)$, which coincides with the two-point function associated to $\omega^M$ on $\radobs\times\radobs$. The very same problem has been encountered already in \cite{Dappiaggi:2011cj, Siemssen:2011gma} for the vector potential and, as in that case, we shall tackle it as follows: We construct an auxiliary two-point function on $\widetilde M$, showing that this enjoys the correct wavefront set condition. Although strictly speaking one should not be entitled to call $\omega^M$ Hadamard, we are still convinced that it deserves this name since the ensuing wavefront set on $\widetilde M$ is built only out of null geodesics which are invariant under conformal transformation. Hence we expect that the singular behaviour of $\omega^M$ is genuinely the same as that of a full-fledged Hadamard state.

\begin{theo}
The state $\omega^M=\omega^\Im\circ\Upsilon:\mathcal{F}(\radobs)\to\mathbb{C}$, defined by the pull-back along $\Upsilon:\mathcal{F}(\radobs)\to\mathcal{F}(\Im^+)$ (see Proposition \ref{b2b-q}) of the state $\omega^\Im$ introduced in Proposition \ref{2pt}, enjoys the following properties:
\begin{enumerate}
\item Its two-point function is the restriction to $\radobs\times\radobs$ of a bi-distribution on $\widetilde M$ whose wavefront set on $\psi(M)$ is of Hadamard from, where $\psi(M)$ is the image of $M$ in $\widetilde M$;
\item It is invariant under the action of all isometries of the bulk metric $g$, that is $\omega^M\circ\alpha_\phi=\omega^M$. Here $\phi:M\to M$ is any isometry and $\alpha_\phi$ represents the action of $\phi$ induced on $\mathcal{F}(\radobs)$ by setting $\alpha_\phi([\epsilon])=[\phi_*\epsilon]$ on the algebra generators $[\epsilon]\in\radobs$;
\item It coincides with the Poincar\'e vacuum on Minkowski spacetime.
\end{enumerate}
\end{theo}

\begin{proof}
To prove $1.$, we follow the same strategy as in \cite{Dappiaggi:2011cj}. Denoting with $\iota$ the embedding of $\Im^+$ into $\widetilde{M}$, we start introducing an auxiliary two-point function:
$$\omega_2^{\widetilde M}(f\otimes f^\prime)\doteq\omega_2^\Im(\iota^*\widetilde G^-(f)\otimes\iota^*\widetilde G^-(f^\prime)),$$
where $\widetilde G^\pm$ are the retarded/advanced fundamental solutions of \eqref{monsterGX} together with \eqref{GXb} and where $f$ and $f^\prime$ are here any pair of smooth and compactly supported test sections generating solutions for \eqref{monsterGX} and \eqref{GXb}, seen as a Green hyperbolic PDE as per Remark \ref{hopedieslastbutonlyIresurrect}. Notice that, per construction $\omega_2^{\widetilde M}$ coincides with the two-point function of $\omega^M$ when we consider initial data for \eqref{monsterGX} together with \eqref{GXb} descending from $\radobs$. By applying on both sides the operator $Q$ and using $\widetilde G^-Q = \mathrm{id}$, for the arbitrariness of $f$ and $f^\prime$ we obtain $\omega_2^{\widetilde M}(Qf\otimes Qf^\prime)=\omega_2^\Im(\iota^*f\otimes\iota^*f^\prime)$. Hence the pull-back of $\omega_2^\Im$ along $\iota:\Im^+\to\widetilde{M}$ has the same wavefrontset as that of $(Q^*\otimes Q^*)\,\omega_2^{\widetilde M}$. If we prove that the latter is the same as that of Hadamard states, we have reached the sought conclusion. This statements descends from the invariance of null geodesics under conformal transformation and from the fact that no null geodesic joining any $x\in M$ to $i^+$ exists \cite[Lemma 4.3]{Moretti:2006ks}. Let us start from the wavefront set of $\omega_2^\Im$ which have been already computed in \cite{Dappiaggi:2008dk}:
$$WF(\omega_2^\Im)=\{(x,x,k,-k)\in T^*(\Im^+\times\Im^+)\setminus\{0\}\;|\;k_u>0\},$$
where $k_u$ is the component of the covector $k$ along the null direction. If we apply the theorem of propagation of singularities to $(Q^*\otimes Q^*)\,\omega_2^{\widetilde M}$ we obtain:
\begin{gather*}
WF(\omega_2^{\widetilde M})=\{(x,y,k_x,-k_y)\in T^*(\widetilde M\times \widetilde M)\setminus\{0\}\;|\;\exists p\in\Im^+,\;q\in T^*_{\iota(p)}\widetilde M\;|\;q_u>0\;\textrm{ such that}\\
\; x,x^\prime\in J^-_{\widetilde{M}}(i^+)\setminus\{i^+\},\;(x,k)\sim (x',k')\sim (\iota(p),q)\},
\end{gather*}
where $\sim$ means that the points are connected by a lightlike geodesic, while the covectors are parallely transported along it. If we add to this result the fact that there does not exist any null geodesic joining a point $x$ in $\psi(M)\subset \widetilde M$ to $i^+$ , {\it c.f.} \cite[Lemma 4.2]{Moretti:2006ks}, then it holds that $\omega_2^{\widetilde M}$ has a wavefront set of Hadamard form is $\psi(M)$.

We prove now $2.$ As a consequence of \cite[Theorem 3.1]{Moretti:2006ks} it suffices to prove the statement for any one-parameter group of isometries $\phi^X_t$ with $t\in\mathbb{R}$ and $X$ a Killing field. Per definition $\omega^M\circ\alpha_{\phi^X_t}=\omega^\Im\circ(\Upsilon\circ\alpha_{\phi^X_t})$. Let $[\epsilon]\in\radobs$ be any generator of $\mathcal{F}(\radobs)$, then $\Upsilon\circ\alpha_{\phi^X_t}([\epsilon])=\Upsilon([\phi^X_{t*}\epsilon])$. Mimicking the same analysis as that of \cite[Proposition 3.4]{Moretti:2006ks}, one gets that $\Upsilon([\phi^X_{t*}\epsilon])=\Pi_{\widetilde\Phi^{\widetilde X}_t}\Upsilon([\epsilon])$, where $\Pi$ is the representation of the BMS group \eqref{BMSrep}, while $\widetilde\Phi^{\widetilde X}_t$ is the action on $\mathcal{F}(\Im^+)$ of a one-parameter group of BMS elements constructed via exponential map from $\widetilde X$, the unique extension of $X$ to $\Im^+$ \cite{Geroch}. Yet Proposition \ref{2pt} entails invariance of $\omega_\Im$ under the action of $\Pi$, from which it descends that $\omega^M\circ\alpha_{\phi^X_t}=\omega^\Im\circ\Pi_{\widetilde\Phi^{\widetilde X}_t}\Upsilon=\omega^\Im\circ\Upsilon=\omega^M$. Notice that point $3.$ is a direct consequence of point $2.$ since, on Minkowski spacetime, $\omega^M$ is a quasi-free and Poincar\'e invariant Hadamard state. Uniqueness of the vacuum yields the sought result.
\end{proof}

\section{\label{sec5}Conclusions}

In this paper we discussed the quantization of linearized gravity on asymptotically flat, globally hyperbolic, vacuum spacetimes within the framework of algebraic quantum field theory. The goal was to construct a distinguished Hadamard state which is invariant under the action of all spacetime isometries. To this end we exploited the existence of a conformal boundary which includes $\Im^+$, future null infinity, a codimension $1$ submanifold on which we defined an auxiliary non-dynamical field theory and an associated $*$-algebra $\mathcal{F}(\Im^+)$. The procedure we followed consists first of all of finding a map which associates to each element of the algebra of fields for linearized gravity a counterpart in $\mathcal{F}(\Im^+)$. This step can be translated into proving that, up to a gauge transformation, each solution of the linearized Einstein's equations admits after a conformal rescaling a smooth extension to $\Im^+$. This operation was thought to be always possible thanks to a suitable gauge fixing, first written by Geroch and Xanthopoulos. We have proven that there exists an obstruction which depends both on the geometry and on the topology of the manifold. Hence, while on certain backgrounds such as for example Minkowski spacetime, all observables admit a counterpart at null infinity, in other scenarios, such as for example axisymmetric backgrounds, this is not the case. We have therefore introduced the notion of radiative observables to indicate those which admit an associated element in $\mathcal{F}(\Im^+)$. These form a not necessarily proper sub-algebra of the algebra of all observables and we have constructed for it a Hadamard state which is invariant under the action of all spacetime isometries. 

We reckon that different follow-up to our analysis are conceivable. From the side of general relativity, the realization of the existence of an obstruction to implement the Geroch-Xanthopoulos gauge, suggests that a critical review of the claimed results about the stability of asymptotic flatness is due. A possible way out would be to find an alternative to the GX-gauge which still allows to extend all spacelike compact solutions of linearized gravity to null infinity. We tried hard to find such alternative but to no avail. If, on the contrary, the obstruction is always present, then it would be interesting to understand whether radiative observables play a distinguished role from a physical point of view. 

From the side of algebraic quantum field theory, our investigation, combined with that of Fewster and Hunt \cite{Fewster:2012bj, Hunt2012} suggests strongly that linearized gravity might behave similarly to the vector potential in electrodynamics with respect to its interplay with both general local covariance and dynamical locality. It might also be interesting to explore other avenues to construct states of interest in physics, particularly following the approach advocated in recent works \cite{Gerard:2012wb, Gerard:2014wb}. More generally, it might be worth trying to extend the notion itself of what is a Hadamard state in the following sense: For gauge theories, states and their two-point functions in particular are defined on suitable gauge equivalence classes of observables, $\clobs$ or $\radobs$ in the case at hand. Yet, in order to claim that a given two-point function is of Hadamard form, one has to show that it can be seen as the restriction of a bi-distribution defined on a whole space of smooth and compactly supported sections so to be able to apply the tools proper of microlocal analysis to check the relevant wavefront sets. 


\section*{Acknowledgements}
We would like to thank Klaus Fredenhagen, Thomas-Paul Hack, Igor Khavkine, Katarzyna Rejzner 
Alexander Schenkel and Daniel Siemssen for useful discussions and comments. S.M. is grateful to Francesco Bonsante and Ludovico Pernazza for useful comments during the early stages of this work. We are greatly indebted to Chris Fewster for useful comments and for pointing us out reference \cite{Hunt2012}  and to Stefan Hollands for pointing us out references \cite{Geroch:1978ur} and \cite{Hollands:2003ie} as well as for enlightening discussions on these papers. The work of C.D.\ has been supported partly by the University of Pavia 
and partly by the Indam-GNFM project {``Influenza della materia quantistica sulle fluttuazioni gravitazionali''}. 
The work of M.B. is supported by a Ph.D. fellowship of the University of Pavia. 


\end{document}